\documentclass[journal, 10pt]{IEEEtran}
\IEEEoverridecommandlockouts
\usepackage{cite}
\usepackage{amsmath,amssymb,amsfonts}
\usepackage{algorithmic}
\usepackage{graphicx}
\usepackage{textcomp}
\usepackage{xcolor}
\usepackage{bm}
\usepackage{subfigure}
\usepackage{bm}

\usepackage{amsthm}
\usepackage{multicol}
\usepackage{authblk}
\usepackage{pifont}
\usepackage{booktabs}

\DeclareMathOperator{\Var}{Var}

\newtheorem{lemma}{\textbf{Lemma}}
\newtheorem{theorem}{\textbf{Theorem}}
\newtheorem{corollary}{\textbf{Corollary}}

\newcommand{\overbar}[1]{\mkern 1.5mu\overline{\mkern-1.5mu#1\mkern-1.5mu}\mkern 1.5mu}

\author{
\IEEEauthorblockN{Daqian Ding, Haorui Li, Yibo Pi, \emph{Member, IEEE}, and Xudong Wang, \emph{Fellow, IEEE}}
\thanks{An earlier version of this paper has been accepted to the IEEE PIMRC 2024~\cite{mn_ige} and will be presented in Spain in September, 2024.}
\thanks{The authors are with the UM-SJTU Joint Institute, Shanghai Jiao Tong
University, Shanghai 200240, China (email: $\{$daqian.ding, haorui.li, yibo.pi$\}$@sjtu.edu.cn, wxudong@ieee.org).}
}

\begin{document}

\title{Interference Graph Estimation for Resource Allocation in Multi-Cell Multi-Numerology Networks: A Power-Domain Approach
}



\maketitle

\begin{abstract}
The interference graph, depicting the intra- and inter-cell interference channel gains, is indispensable for resource allocation in multi-cell networks. 
However, there lacks viable methods of interference graph estimation (IGE) for multi-cell multi-numerology (MN) networks. To fill this gap, we propose an efficient power-domain approach to IGE for the resource allocation in multi-cell MN networks. Unlike traditional reference signal-based approaches that consume frequency-time resources, our approach uses power as a new dimension for the estimation of channel gains. By carefully controlling the transmit powers of base stations, our approach is capable of estimating both intra- and inter-cell interference channel gains. As a power-domain approach, it can be seamlessly integrated with the resource allocation such that IGE and resource allocation can be conducted simultaneously using the same frequency-time resources. We derive the necessary conditions for the power-domain IGE and design a practical power control scheme. We formulate a multi-objective joint optimization problem of IGE and resource allocation, propose iterative solutions with proven convergence, and analyze the computational complexity. Our simulation results show that power-domain IGE can accurately estimate strong interference channel gains with low power overhead and is robust to carrier frequency and timing offsets.
\end{abstract}


\section{Introduction}
With the goal of supporting diverse services, the 5G new radio (NR) systems introduce multi-numerology (MN) that allows flexible selection of frame structure based on service requirements, which is also envisioned to support immersive communications in 6G~\cite{mn_6g}. Frames with large subcarrier spacing (SCS) may be employed for users with strict latency requirements and frames with small SCS are preferable for users with highly frequency-selective channels. The flexibility of multi-numerology enables 5G to simultaneously support services of different requirements. However, the use of multiple numerologies also incurs the inter-numerology interference (INI) between different numerologies. In contrast to single-numerology networks, the presence of INI greatly complicates the interference management in MN networks. Each user equipment (UE) may experience not only interference from neighboring cells on its assigned frequency band but also INI from other frequency bands in both the serving and neighboring cells. It is thus critical to manage INI for the performance of MN networks. 

Effective INI management can be achieved by optimizing the MN system from different angles. From the perspective of system design, many design parameters, e.g., waveform, guard band, and cyclic prefix (CP), can be considered. New waveforms have been proposed for several typical subband filtered multi-carrier (SFMC) based systems to reduce the out-of-band emission level, including filter-band multi-carrier~\cite{fbmc, ofdm-fbmc}, universal filtered multi-carrier~\cite{ufmc}, generalized frequency division multiplexing~\cite{gfdm}, and filtered OFDM~\cite{async_f_ofdm}. Large guard bands are shown to effectively mitigate INI at the cost of spectral efficiency~\cite{guard_band, f_ofdm}. When the system design is fixed, we can further mitigate INI with resource allocation~\cite{ini_network_slicing, ini_aware_scheduling, multi_numerology_b5g}, which emphasizes the proper allocation of power and frequency-time resources to guarantee certain service requirements. In multi-cell MN networks, both the INI within the cell (intra-cell INI) and the INI between neighboring cells (inter-cell INI) need to be considered.



The interference graph, depicting the intra- and inter-cell interference channel gains, enables us to estimate link interference and the signal-to-interference-plus-noise ratios (SINRs) of UEs under different resource allocation schemes. It is thus a powerful tool to guide resource allocation for interference management~\cite{ige_resource_alloc_1, ige_resource_alloc_2, ige_resource_alloc_3, interference_graph_def}. Existing works on resource allocation in multi-cell MN networks typically assume knowing a certain form of interference graph and focus on developing resource allocation schemes. However, interference graph estimation (IGE) is a challenging and heavy measurement task, because it requires estimating the channel gains between each pair of UE and interferer. In a single-numerology network of $N$ cells, each UE needs to consider interference from each cell, equating to $O(N)$ inter-cell interference links. In MN networks, if each cell uses $M$ numerologies and that both intra- and inter-cell INIs exist, each UE may experience INI from $O(MN)$ sources of interference. 
It is therefore of paramount importance to efficiently estimate the interference graph. Efficient power-domain approaches have been proposed to estimate the interference graph in millimeter-wave backhaul networks ~\cite{power_domain_ige_globecom, power_domain_ige_twc} and have been integrated with concurrent flooding in BLE multi-hop networks~\cite{power_domain_ige_ewsn}. However, these efforts focus on the power-domain approaches in the time domain for single-numerology systems. There still lacks an efficient approach to estimate interference graph in the frequency domain for multi-cell MN networks.

INI models have been derived and used to estimate the interference in various MN systems. In \cite{windowed_ofdm, sfmc, combining_massive_mimo}, the closed-form expressions for intra-cell INI have been derived in the windowed OFDM, SFMC, and massive MIMO-OFDM systems for interference cancellation at the receiver, respectively. In~\cite{mimo_ofdm_transmit}, the expression for INI has been used to suppress intra-cell INI with precoding at the transmitter in the downlink of massive MIMO-OFDM systems. These theoretical expressions for INI assume the absence of system imperfections, e.g., time synchronization error and phase noise, and may not be applicable when practical issues are considered~\cite{sfmc, phase_noise}. In the resource allocation problems for inter-cell MN networks, intra- and inter-cell INIs need to be modelled to compute SINR. In \cite{multi_numerology_b5g}, the intra-cell INI is simply modelled without involving the channel responses. In \cite{multi_numerology_comp}, inter-cell INI is modelled as the product of the intra-cell INI and the inter-cell channel responses. These simplified models cannot reflect the INI in real systems and may significantly affect the performance of resource allocation when used in practice. 

Compared to the model-based approaches above, reference signal (RS)-based approaches are capable of measuring interference channel gains in single-numerology systems, but struggle to estimate the INI channels in MN systems due to the mismatch in numerology. In single-numerology systems, the most practical approach to measuring inter-cell interference is introduced in 4G LTE-A using the CSI-IM reference signal~\cite{csi_im}. Specifically, in order to measure the inter-cell interference from a BS to a UE on a specific resource element (RE), other BSs should send no signal on that RE. The signal UE receives on the RE will then be the interference from the BS. This process is repeated for each BS to construct the interference graph from BSs to UEs. However, this approach requires symbol-level synchronization between BSs for accurate estimation of the interference channels. With the increasing subcarrier spacing in 5G NR, the CP duration decreases to the millisecond level~\cite{TS38211}, making this approach susceptible to synchronization error~\cite{time_sync}. Furthermore, this approach consumes frequency-time resources solely for measurement purposes, and frequent measurements may incur non-negligible overhead.
Apart from the above issues, the numerology mismatch between the interferer and the receiver in MN networks prevents the receiver from estimating the interference channel gains using reference signals. Consequently, reference signals are limited to estimate the communication channel, and it remains unclear how to measure the INI channels in multi-cell MN networks using reference signals.

Considering the heavy measurement overhead of IGE, machine learning approaches have been proposed to build a direct mapping from network attributes to resource allocation decisions, without explicitly measuring interference. Geographic information of nodes has been used as input to learn the link scheduling decisions from a large scale of layouts in device-to-device networks~\cite{spatial_dl}. With graph embedding, the same task for link scheduling can be achieved using hundreds of times fewer layouts for training~\cite{graph_link_schedule}. Likewise, deep learning models can also be used to improve other performance metrics, e.g., age of information (AoI)~\cite{aoi_dl}. Although these machine learning techniques do not need explicit interference measurement, their dependence on static network attributes for decision-making constrains their ability to adapt to rapidly changing network conditions, e.g., fading.

\begin{table}[h]\label{table:comparison}
\centering
\caption{Comparison among approaches for IGE}
\renewcommand{\arraystretch}{1.1}
\begin{tabular}{lccc}
\toprule
 & Our & RS-based & Model-based \\
\midrule
Accurate IGE & \ding{52} & \ding{55} & \ding{55}\\
\hline
Robust to time sync error & \ding{52} & \ding{55} & \ding{52} \\
\hline
Adaptive to environment & \ding{52} & \ding{52} & \ding{55}  \\
\hline
No extra freq-time resource & \ding{52} & \ding{55} & \ding{52}\\
\hline
No extra power consumption & \ding{55} & \ding{52} & \ding{52}\\
\bottomrule
\end{tabular}
\end{table}

Our focus in this paper is on the estimation of interference graph for resource allocation in the downlink of OFDMA-based multi-cell MN networks.
Specifically, we want to 1) estimate the interference graph including both the intra- and inter-cell interference channels and 2) use the estimated interference graph to improve the energy efficiency of the system. Inspired by the IGE approach in~\cite{ige_mmwave_backhauling}, we propose a power-domain IGE approach that estimates the interference graph by manipulating the transmit power of BSs for multi-cell MN networks, such that the same frequency-time resources can be simultaneously used for measurement and data transmission tasks. 
Our core insight is that the expected receive power of a UE at a time slot can be written as a linear combination of the products of transmit powers of BSs and the channel gains. By combining the received powers of the UE at different time slots, a group of equations for the same set of channel gains can be obtained to derive a unique solution for IGE. With the estimated intra- and inter-cell channel gains, we can formulate a resource allocation problem aiming to maximize energy efficiency and estimate interference channel gains simultaneously. Table \ref{table:comparison} summarizes the desired properties for IGE approaches and compare the performance between our approach and existing ones, where the model-based approach models the INI with fixed properties, e.g., distance.

In summary, our contributions in this paper are as follows.
\begin{itemize}
    \item We propose a power-domain approach for interference graph estimation by controlling the transmit powers of BSs, such that interference graph estimation and resource allocation can be conducted simultaneously using the same frequency-time resources.
    \item We design a practical power control scheme that achieves millisecond-scale updates of channel gain estimates and is robust to timing and carrier frequency offsets. We conduct the error analysis for IGE to characterize the impact of the major sources of errors and derive an upper bound for the estimation error.
    \item We formulate the joint optimization problem of resource allocation and IGE for multi-cell multi-numerology systems. We propose an interative heuristic sollution with provable convergence and analyze the computational complexity.
    \item We demonstrate with simulations that our approach can accomplish resource allocation and IGE simultaneously with low overhead in energy efficiency.
\end{itemize}

Note that compared to our previous conference paper~\cite{mn_ige}, we have added substantial materials in this paper, including a practical power control scheme, the bounds of channel gain estimation errors, a multi-objective formulation of the joint optimization problem, the convergence and complexity analysis of the propose solution, and more simulation results. The rest of this paper is organized as follows. Section \ref{sec:system_model} presents the system model. Section \ref{section_interference} presents a practical power-domain approach to IGE and conducts the error analysis. A joint optimization problem is formulated and solved in Section \ref{sec:joint_opt} and the performance is evaluated in Section~\ref{sec:perf_eval}. Section~\ref{sec:conclusion} concludes the paper.

\section{System Model}
\label{sec:system_model}

\begin{figure}[t]
    \centering
    \subfigure[Network architecture]{
        \label{fig:sys_arch}
        \includegraphics[scale=0.37]{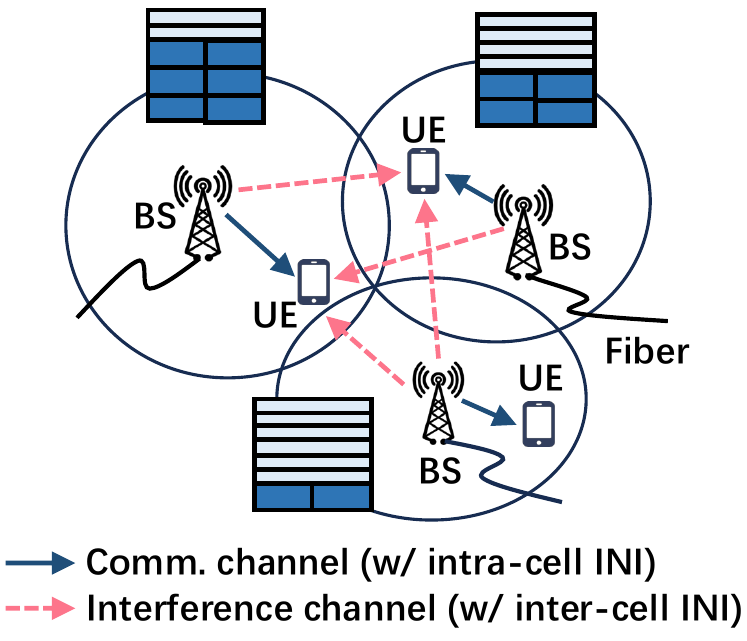}
    }
    \subfigure[Frame structure]{
        \label{fig:frame_struct}
        \includegraphics[scale=0.38]{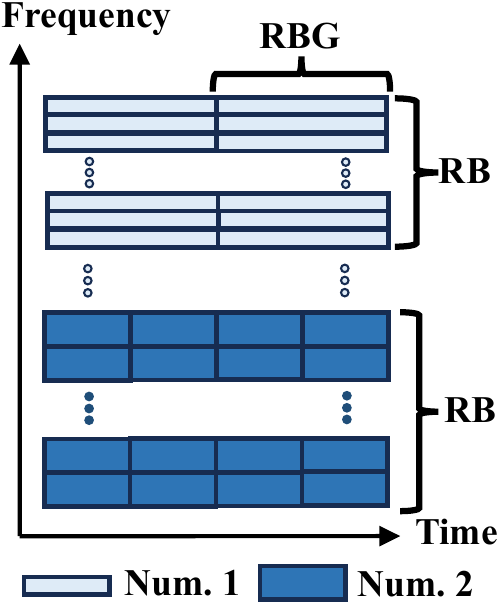}
    }
    \vspace{-0.5em}
    \caption{A multi-cell multi-numerology network}
    \label{fig:error_bound}
    \vspace{-1em}
\end{figure}

\subsection{Multi-Cell Multi-Numerology Networks}
We consider the downlink of a multi-cell OFDMA network with $K$ BSs connected to the core network via a fiber backhaul, where the fiber backhaul enables the BSs to achieve low-latency message exchange for joint power control (Figure \ref{fig:sys_arch}). Each UE is associated with one BS, and BSs serve each of its associated UEs with one numerology and use multiple numerologies to provide diverse access services. \textcolor{black}{For simplicity, we assume fixed but different numerology schemes across frequency bands over time, with each frequency band using the same numerology.} For spectrum efficiency, each BS re-uses the same frequency-time resources and UEs receive the inter-cell interference from neighboring BSs. Each UE experiences multiple sources of interference: 1) intra-cell INI in the communication channel, due to the non-orthogonality between different numerologies used by the BS, 2) inter-cell co-channel interference from neighboring cells in the same frequency, 3) inter-cell INI from neighboring cells in other frequency bands using different numerologies, and 4) inter-carrier interference due to carrier frequency offset. To mitigate interference, BSs efficiently coordinate resource allocation such that the demands of UEs are satisfied. In this paper, we consider the timing offset (TO) and carrier frequency offset (CFO) of the received signal at the UEs. The CFO is caused by the oscillator mismatch in frequency~\cite{interference_cfo} and the TO is caused by the time synchronization errors between BSs as well as the differences in propagation delay from BSs to the UE~\cite{timing_offset}. 

Let $\mathcal{I} = \{0,\dots,I\}$ be the total numerology set and $\mathcal{I}_k \subseteq \mathcal{I}$ be the numerology set used by BS $k$, where numerology $0$ has the narrowest SCS. Let $\Delta f^{i}$ be the SCS of numerology $i$, and $N^{i}$ and $N^{i}_{cp}$ be the lengths of the OFDM symbols and CPs corresponding to numerology $i$. Following the standard design in 5G NR~\cite{release_17}, we relate the SCSs, the lengths of OFDM symbols and CPs between numerologies as
\begin{equation}\label{eq:numerology_scheme}
    \frac{N^{i}}{N^{0}} = \frac{N_{cp}^{i}}{N_{cp}^{0}} = \frac{\Delta f^{0}}{\Delta f^{i}} = \frac{1}{2^i}.
\end{equation}
\textcolor{black}{We assume that the CP lengths of all numerologies are greater than the delay spread to avoid inter-symbol interference (ISI).}
Figure \ref{fig:frame_struct} shows the frame structure of the resource grid with two numerologies following the above relations. Each resource block (RB) consists of 12 subcarriers and lasts for one OFDM symbol time. Consecutive RBs are combined into RB groups (RBGs), which are the least common multiple of the symbol durations across all numerologies. Since the RBGs of different numerologies are aligned in time, we can flexibly schedule resources at the RBG level. 





\subsection{Channel Model}


The $n$-th sample of the time-domain transmitted signal from BS $k$ is denoted as $x_k[n]$, which is the combined signal of the numerologies in $\mathcal{I}_k$.  We can express $x_k[n]$ as $x_k[n]= \sum_{i\in\mathcal{I}_k} x^{i}_k[n]$, where $x^{i}[n]$ is the signal of numerology $i$. Let $X_k^{i}[m]$ be the symbol transmitted over the $m$-th subcarrier using numerology $i$ and $\mathcal{Z}^{i}_k$ be the set of RBs using numerology $i$, where these RBs are not necessarily contiguous. 
Setting $X_k^{i}[m]$ to zero for subcarriers not in the RBs in $\mathcal{Z}^{i}_k$, we can simply write $x_k^{i}[n]$ as
\begin{equation*}
    x_k^{i}[n] = \frac{1}{\sqrt{N^{i}}}\sum_{m=0}^{N^{i}-1}X_k^{i}[m]e^{j2\pi mn/N^{i}}.
\end{equation*}

We denote the CFO of BS $k$ as $\omega_k$, normalized to the subcarrier spacing, and the synchronization error as $\zeta_k$. Suppose the channel from BS $k$ to UEs is an independent $L$-tap multipath fading channel. The $n$-th sample in the time-domain received signal of UE $z$ can be expressed as
\begin{equation}\label{eq:received_signal}
    y_z[n] = \sum_{k\in \mathcal{K}} e^{j2\pi n\omega_k/N^{i}_{k,z}}\sum_{l=0}^{L - 1} h_{k,z}^{(l)} x_k[n - l-\zeta_k] + v[n],
\end{equation}
where $N^{i}_{k,z}$ is the number of subcarriers of UE $z$ in BS $k$ using numerology $i$, $h^{(l)}_{k,z}$ is the $l$-th channel tap from BS $k$ to node $z$, and $v[n]$ is the zero-mean additive Guassian noise with variance $\sigma^2_v$. 




\section{Interference Graph Estimation: \\ A Power-Domain Approach}
\label{section_interference}

\subsection{Estimating the Interference Graph with Power Control}

Let $T_{LCM}$ be the symbol duration of numerology 0, i.e., the duration of RGBs, and $T_i$ be the symbol duration of numerology $i$. Based on the numerology properties in Eq. (\ref{eq:numerology_scheme}), we then have $T_{LCM} = 2^{i}T_i$, where numerology $0$ has the narrowest SCS. We first derive the power linearity when CP exceeds TO and then extend the derivation to cases where CP is insufficient, i.e., shorter than TO, in Section~\ref{sec:error_analysis}. It should be noted that TO includes both the time synchronization error and delay spread.

\begin{lemma}\label{lemma:linearity_slow_moving}
For the downlink of the OFDMA-based multi-cell multi-numerology systems, if CP exceeds TO and $\mathbb{E}[X_k^{i}[m]] = 0$ for each subcarrier, the scaled expected receive power of UE $z$ on subcarrier $d$ is a linear combination of the products of the equivalent channel gains and the scaled expected transmit powers of BSs on subcarriers, i.e.,
\begin{align*}
2^i\hat{p}_{z, rx}^i[d] = \sum_{k \in \mathcal{K}}\sum_{j\in\mathcal{I}_k}\sum_{m=0}^{N^j-1} 2^j\hat{p}_{k, tx}^j[m] s_{(k,m),(z,d)}^{j \rightarrow i} + \hat{V}_z[d],
\end{align*}
where $\hat{p}_{z,rx}^i[d] = \mathbb{E}[|Y_z^i[d]|^2]$ is the receive power of UE $z$ on the $d$-th subcarrier using numerology $i$, $\hat{p}_{k, tx}^i[m] = \mathbb{E}[|X_k^j[m]|^2]$ is the transmit power of BS $k$ on the $m$-th subcarrier using numerology $j$, $\hat{s}_{(k,m),(z,d)}^{j \rightarrow i}$ is the equivalent channel gain from the $m$-th subcarrier of BS $k$ to the $d$-th subcarrier of UE $z$, and $\hat{V}_z[d]$ is the noise power on the $d$-th subcarrier of UE $z$.
\end{lemma}

\begin{IEEEproof}
Please refer to Appendix \ref{section:linearity_slow_moving}.
\end{IEEEproof}



Lemma \ref{lemma:linearity_slow_moving} indicates the relation between power and channel gains at the subcarrier level. Considering that frequency-time resources in OFDMA-based systems are allocated at the resource block (RB) level~\cite{resource_alloc_granularity}, we want to aggregate subcarriers into RBs of $B$ subcarriers, where subcarriers in the same RB are assumed using identical transmit power. Let $\mathcal{RB}_k$ be the set of RBs used by BS $k$ and $\mu_d$ be the numerology of RB $d$. By summing up the receive powers on subcarriers in a RB, we can have the following corollary.

\begin{corollary} \label{corollary1}
The scaled expected receive power of a UE on the $d$-th RB is a linear combination of the products of the equivalent channel gains and the scaled expected transmit powers of BSs on RBs, i.e.,
\begin{equation*}
    2^{\mu_d}p_{z}^{rx}[d] = \sum_{k \in \mathcal{K}}\sum_{l\in\mathcal{RB}_k} 2^{\mu_l}p_{k}^{tx}[l] s_{(k,l),(z,d)} + V_z[d],
\end{equation*}
where $p_{z}^{rx}[d] = \sum_{j=B(d-1)+1}^{Bd} \hat{p}_{z, rx}^{\mu_d}[j]$ is the receive power of UE $z$ on RB $d$, $p_{k}^{tx}[l] = B\hat{p}_{k, tx}^{\mu_l}[B(l-1)+1]$ is the transmit power on the $l$-th RB of BS $k$, $s_{(k,l),(z,d)} = \frac{1}{B^2}\sum\limits_{r=B(l-1) + 1}^{Bl}\sum\limits_{j=B(d-1) + 1}^{Bd}\hat{s}_{(k,r),(z,j)}^{\mu_l \rightarrow \mu_d}$, and $V_z[d] = \sum_{j=B(d-1) + 1}^{Bd}\hat{V}_z[j]$.
\end{corollary}

Based on Corollary \ref{corollary1}, we can derive the sufficient condition for achieving IGE with power control. Let the total number of RBs involved in the multi-cell MN networks be $n_{tot}$ with indices from 1 to $n_{tot}$.

\begin{theorem}\label{theorem:full-rank}
The equivalent channel gains on different RBs from BSs to UEs can be uniquely determined by controlling the scaled transmit powers of BSs on RBs over time such that
\begin{equation}
    rank(\mathbf{P}) = n_{tot},
\end{equation}
where $\mathbf{P} = \left[\bm{p}_1^{tx},\dots, \bm{p}_{n_{tot}}^{tx}\right]$ is the scaled transmit power matrix of BSs, $\bm{p}_i^{tx} = \left[2^{\mu_i}p_{i}^{tx}[j],\dots, 2^{\mu_i}p_{i}^{tx}[j + n - 1]\right]^T$ includes the scaled transmit powers on the $i$-th RB from time $j$ to $j+n-1$, and $\mu_i$ is the numerology of RB $i$.
\end{theorem}

\begin{IEEEproof}
Let $\bm{p}^{rx}_i = \left[2^{\mu_i}p_i^{rx}[j],\dots, 2^{\mu_i}p_i^{rx}[j + n - 1]\right]^T$ be the scaled receive powers of a UE on RB $i$ over time. Let $s_{d,i}$ denote the interference channel gain between RB $d$ of the BS and RB $i$ of the UE. According to Corollary~\ref{corollary1}, we have
\begin{equation}\label{eq:linearity_matrix_form}
    \bm{p}^{rx}_i = \mathbf{P}\bm{s}_i + \bm{v}_i,
\end{equation}
where $\bm{s}_i = \left[s_{1,i},\dots, s_{n_{tot}, i}\right]^T$ and $\bm{v}_i = [V_i[j],\dots,V_i[j+n-1]]^T$ is the vector of noise power. If the transmit power matrix $\mathbf{P}$ is full-rank, we can then have a unique solution for $\bm{s}_i$. This holds for all UEs and enables us to estimate the interference channel gain between each pair of BS and UE at the RB level.
\end{IEEEproof}

Note that our equivalent channel gains include only the magnitude, which is sufficient for resource allocation; however, they lack the phase information necessary for demodulation.

\subsection{Practical Power Control Scheme}

\begin{figure}[!t]
    \centering
        \includegraphics[scale=0.45]{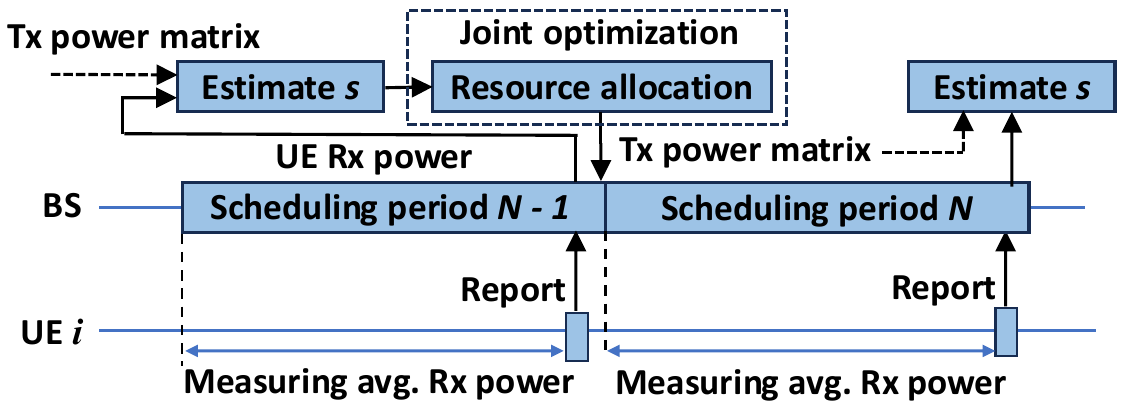}
    \caption{High-level workflow of the power control scheme}\label{fig:overview}
    \vspace{-1em}
\end{figure}

\begin{figure}[!t]
    \centering
        \includegraphics[scale=0.45]{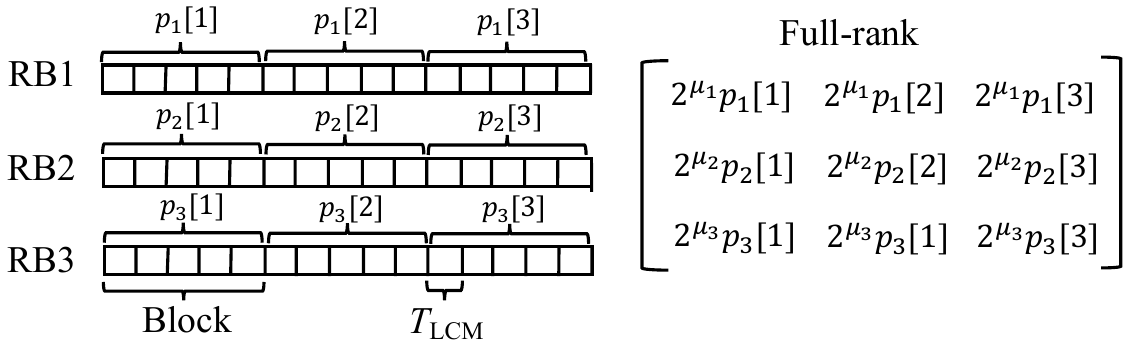}
    \caption{Practical power control scheme for interference graph estimation}\label{fig:power_control_scheme}
    \vspace{-1em}
\end{figure}

\noindent\textbf{High-Level Workflow.} As shown in Fig.~\ref{fig:overview}, resource allocation is conducted periodically. At the end of each scheduling period, UEs report their average receive powers back to the BSs. Combining the average receive powers of UEs in the current period and the transmit powers of interfering BSs determined at the end of the previous period, each BS can estimate the equivalent interference channel gains at the RB level for UEs using Eq.~(\ref{eq:linearity_matrix_form}). Then, BSs estimate the equivalent interference channel gains and solve a joint optimization problem to allocate transmit powers for the next period. This power allocation not only maximizes energy efficiency but also ensures the full rank of the transmit power matrix for IGE.

\noindent\textbf{Forming the Full-Rank Transmit Power Matrix.} We form the full-rank transmit power matrix in Theorem~\ref{theorem:full-rank} by varying the transmit powers of BSs on RBs over time. Fig.~\ref{fig:power_control_scheme} shows an example of the practical power control scheme involving three RBs, each of which may interfere adjacent RBs with INI and ICI. To measure the interference channel gains between these RBs and a UE, the transmit power on each RB changes for every block, where each \emph{block} lasts for multiple $T_{LCM}$'s. The transmit power on each RB is scaled to form a full-rank matrix such that a unique solution can be obtained for the equivalent channel gains.

\noindent\textbf{Approximating the Expected Transmit Power.}
To estimate equivalent channel gains by Theorem~\ref{theorem:full-rank}, we want the actual transmit power to approximate the expected transmit power. This requires us to send a sufficient number of samples on each RB to ensure that the average transmit power on each RB closely approximates the expected one. 
Let $X_{k,j}[m]$ and $n_k[m]$ be the $j$-th frequency-domain symbol and the number of samples on the $m$-th subcarrier of BS $k$, $X^{I}_{k,j}[m]$ and $X^{R}_{k,j}[m]$ be the imaginary and real parts of $X_{k,j}[m]$, and $X^{I, \max}_{k,j}[m]$ be the maximum imaginary part of constellation symbols. Let $\mu_l = \sum_{m=(l-1)B+1}^{lB}\mathbb{E}[|X_{k,j}[m]|^2]$ be the expected transmit power of RB $l$, $\hat{\mu}_l = \frac{1}{Bn_k[m]}\sum_{m=(l-1)B+1}^{lB}\sum_{j=1}^{n_k[m]}|X_{k,j}[m]|^2$ be the average transmit power of RB $l$, and $0 \leq \delta \leq 1$ be the approximation error. We have the following theorem to approximate $\mu_m$ with $\hat{\mu}_m$ at a high confidence level.

\begin{theorem}\label{theorem:bennet}
Assume that modulated symbols are chosen with equal probability from the constellation with $\mathbb{E}[X_{k,j}[m]] = 0$ and $\Var(X^{I}_{k,j}[m]) = \Var(X^{R}_{k,j}[m]) = \sigma_m^2$. The average transmit power can be bounded as
\begin{equation*}
    \mathbb{P}[|\hat{\mu}_l - \mu_l| \geq \mu_l\delta] \leq \exp\left\{\frac{-n_{k,l}\Tilde{\sigma}^2_m}{b_m^2}F\left(\frac{2b_m\sigma_m^2\delta}{\Tilde{\sigma}^2_m}\right)\right\},
\end{equation*}
where $F(t) = (1+t)\log(1+t)$, $\Tilde{\sigma}^2_m = 2\mathbb{E}[|X^I_{k,j}[m]|^4] - 2\sigma_m^4$, $b_m = 2|X^{I,\max}_{k,j}[m]^2 - \sigma_m^2|$, and $n_{k,l} = \sum_{m=(l-1)B+1}^{lB} n_k[m]$. When the square $M$-QAM modulation with $M$ = $\{4$, $16$, $64$, $\dots\}$ is used, we have
\begin{equation*}
    \Tilde{\sigma}^2_m = \frac{8d_m^4(M^2 - 5M + 4)}{45} \text{ and }
    b_m = \frac{4d_m^2|M - 3\sqrt{M} + 2|}{3},
\end{equation*}
where $d_m$ is the minimum distance between the constellation symbols for subcarrier $m$.  
\end{theorem}

\begin{proof}
    Let $\Tilde{X} = |X_{k,j}[m]|^2 - \sigma_m^2$ be a zero-mean random variable such that $\Var(\Tilde{X}) = \Tilde{\sigma}_m^2$ and $|\Tilde{X}| \leq b_m$. We know that $\mu_m = \mathbb{E}[|X_{k,j}[m]|^2] = 2\mathbb{E}[|X^I_{k,j}[m]|^2] = 2\sigma_m^2$. The bound can be obtained with the Bennett's inequality. When $M$-QAM is used, $X^{I,max}_{k,j}[m] = (\sqrt{M} - 1)d_m$ and $\sigma_m^2 = \frac{M-1}{3}d_m^2$.
\end{proof}

Note that for modulations with symbols of equal energy, e.g., QPSK and BPSK, the average transmit power $\Tilde{\mu}_l$ is always equal to $\mu_l$, regardless of the number of transmitted symbols. Given the approximation error $\delta$, we can calculate the block length based on Theorem~\ref{theorem:bennet}.


\noindent\textbf{Speeding up the IGE Process.}
In cellular networks, reference signals for channel gain estimation are typically sent from BSs to UEs every tens of milliseconds to track channel dynamics~\cite{rs_periodicity}. We therefore want our power-domain approach for IGE to achieve millisecond-scale update of channel gain estimates. Theorem~\ref{theorem:full-rank} measures the interference between each pair of RBs and thus requires the number of blocks for IGE to be no less than the total number of RBs to form a full-rank transmit power matrix. However, when the number of RBs in a multi-cell network is large, the IGE process will become time-consuming, making it unable to track the channel dynamics in a timely manner. To accelerate the IGE process, our core insight is to reduce the number of interferers for each RB by neglecting the weak ones. Specifically, for each RB, we consider only the RBs with strong interference and neglect the ones with weak interference. As indicated in~\cite{ini_spectral_dist}, the INI between two RBs decrease sharply with their spectral distance. For each RB, we can distinguish its strong and weak interfering RBs based on their spectral distances. To characterize how fast interference channel strength drops with the spectral distance, we derive the following theorem.
\begin{theorem}\label{theorem:spectral_distance}
The strength difference between two equivalent channel gains from RBs of different numerologies to an UE is inversely proportional to the square of the number of subcarriers in a RB.
\end{theorem}
\begin{proof}
    Please refer to Appendix~\ref{sec:proof_spectral_distance}.
\end{proof}
For a typical subcarrier count of 12 in a RB, a spectral distance of one RB may result in the interference channel strength to drop more than 20dB, consistent with our simulation results in Section~\ref{sec:perf_eval}. Thus,  considering only the most adjacent RBs is a reasonable tradeoff between speed and accuracy for IGE. More details on the IGE performance will be covered in Section~\ref{sec:perf_eval}.

\begin{figure}[!t]
    \centering
        \includegraphics[scale=0.45]{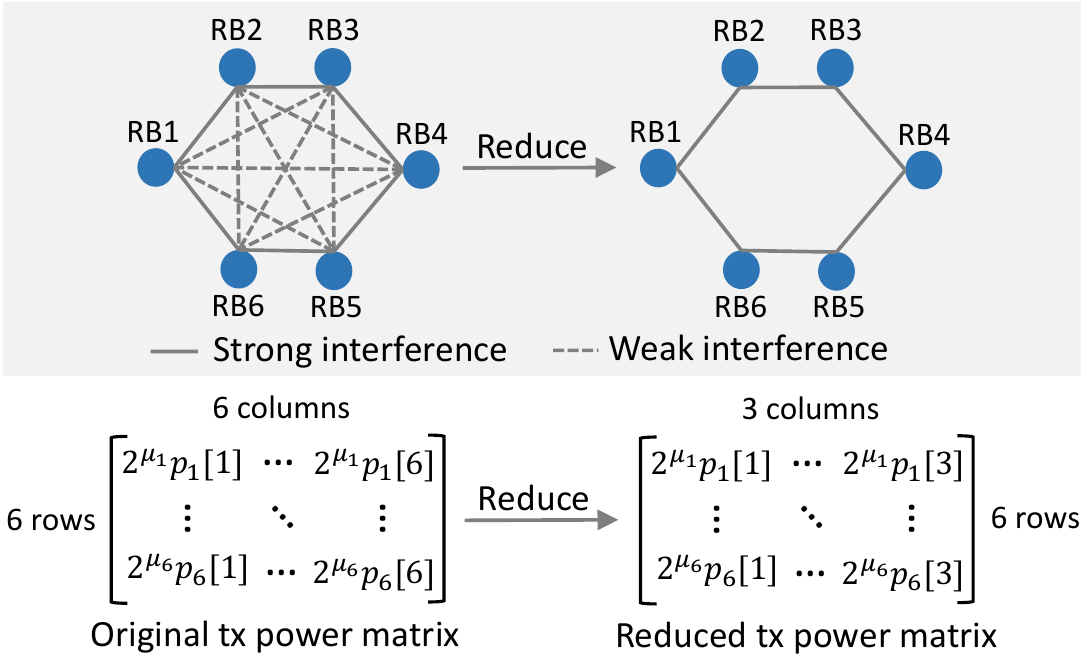}
    \caption{Speeding up IGE by neglecting weak interference between RBs}\label{fig:graph_reduction}
    \vspace{-1em}
\end{figure}

Fig.~\ref{fig:graph_reduction} shows an example involving 6 RBs, where interference exists between each pair of RBs. If the interference between each RB pair is considered, we need to construct a $6 \times 6$ transmit power matrix, which takes 6 blocks to complete the IGE process. To reduce the time for IGE, we only consider the RBs with strong interference. In this example, each RB is assumed to be significantly affected by only two other RBs. As a result, three blocks (two for the interfering RBs and one for the RB itself) are sufficient for the IGE process. The dimension of the original transmit power matrix is reduced to $6 \times 3$. The reduced transmit power matrix has a maximum rank of 3, which cannot guarantee a unique solution for the channel gains. Instead, we require the transmit power vectors of each RB and its interfering RBs to form a full-rank matrix for IGE. For a RB in a $m$-cell network, if it is affected by $n_i$ strong interfering RBs in cell $i$, the total number of interfering RBs is $d = \sum_{i=1}^m n_i$. In other words, it takes $d + 1$ blocks for the IGE process to complete. Using a larger $d$ improves the accuracy of IGE by diminishing the influence of the neglected weak interfering RBs, but also lengthens the IGE process.

For simplicity, we consider $d$ interfering RBs for each RB. We want to construct a $n_{tot}\times d$ transmit power matrix such that each $d \times d$ submatrix formed by the RBs and their interfering RBs has full rank, where $n_{tot} \geq d$ is the total number of RBs. Such a transmit power matrix can be proved to exist as follows.
Given a group of mutually interfering RBs, it is feasible to construct a reduced transmit power matrix to measure the equivalent interference channel gains from these RBs to UEs.
As indicated by Theorem~\ref{theorem:full-rank}, the equivalent interference channels from RBs to a UE are measurable if the transmit power matrix is full rank. Let each RB have an equal number of interfering RBs, denoted as $d$. We can construct a $n_{tot} \times d$ reduced transmit power matrix with $n_{tot} \geq d$ as follows. We first construct a $d\times d$ full-rank transmit power matrix and denote the $i$-th row as $\bm{r}_i$. The $(i+1)$-th row can be constructed as $\bm{r}_{d+1} = \sum_{i=1}^d l_i \bm{r}_i$, where all $l_i$'s are nonzero. It is easy to see that $\bm{r}_{d+1}$ is linearly independent with any $d - 1$ vectors among $\bm{r}_1, \dots, \bm{r}_d$. In other words, any $d$ vectors among $\bm{r}_1,\dots,\bm{r}_{d+1}$ can form a full-rank matrix. Similarly, we can construct $r_{d+2}$ with a linear combination of $d$ vectors among $\bm{r}_1,\dots, \bm{r}_{d+1}$. Following this approach, we can generate $n_{tot}$ vectors to compose a $n_{tot}\times d$ reduced transmit power matrix whose $d\times d$ submatrices are all full-rank.


\begin{figure}[!t]
    \centering
        \includegraphics[scale=0.45]{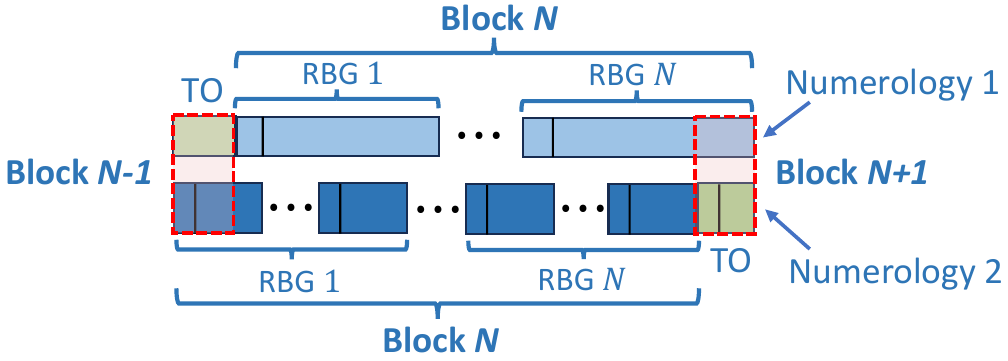}
    \caption{Inter-block interference when TO exceeds CP, where TO is mainly due to the time synchronization errors between BSs.}\label{fig:to_error}
    \vspace{-1em}
\end{figure}

\subsection{Error Analysis for IGE}
\label{sec:error_analysis}

Let $\bm{\overbar{p}}_i^{rx}$ and $\Delta \bm{\overbar{p}}_i^{rx}$ be the average receive power and the difference between the expected and average receive powers of UE $i$. We can rewrite Eq.~(\ref{eq:linearity_matrix_form}) as
\begin{equation*}
\mathbf{P}\bm{s}_i = \bm{\overbar{p}}_i^{rx} + \Delta \bm{\overbar{p}}_i^{rx} + \bm{v}_i.
\end{equation*}
Let $\bm{\hat{s}}_i$ be the estimated equivalent channel gain vector of UE $i$. We can simply estimate $\bm{\hat{s}}_i$ as
\begin{equation}\label{eq:least_square}
\bm{\hat{s}}_i = \mathbf{P}^{-1}\bm{\overbar{p}}_i^{rx}.
\end{equation}
Based on the perturbation theory in~\cite{chandrasekaran1991perturbation}, we can bound the estimation error as
\begin{equation}\label{eq:error_bound}
    \frac{\lVert\bm{\hat{s}}_i - \bm{s}_i\rVert}{\lVert \bm{s}_i \rVert} \leq \kappa(\mathbf{P})\left(\frac{\lVert\Delta \bm{\overbar{p}}_i^{rx}\rVert}{\lVert\bm{\overbar{p}}_i^{rx}\rVert} + \frac{\lVert \bm{v}_i\rVert}{\lVert\bm{\overbar{p}}_i^{rx}\rVert}\right),
\end{equation}
where $\kappa(\mathbf{P})$ is the condition number of $\mathbf{P}$. This indicates that lowering the estimation error requires us to reduce $\kappa(\mathbf{P})$ and $\Delta \bm{\overbar{p}}_i^{rx}$. The condition number of $\mathbf{P}$ can be reduced by carefully choosing the transmit power as in Section~\ref{sec:joint_opt}, while $\Delta \bm{\overbar{p}}_i^{rx}$ can be reduced by using a larger block based on Theorem~\ref{theorem:bennet}.
When TO is less than CP, the major source of error for the receive power is due to the limited sample size of received symbols. When TO is larger than CP, we further need to consider the inter-block interference as shown in Fig.~\ref{fig:to_error}, where RBs using different numerologies are grouped into blocks of the same length and the transmit powers on these RBs vary across blocks for IGE. Inter-block interference causes the leading or trailing symbols in a block to receive signals sent with a different transmit power, which results in another source of error in receive power. We analyze the impact of this error to demonstrate the robustness of our approach to TO.

Following the derivation in Appendix~\ref{section:linearity_slow_moving}, we can derive a similar linear relation between the transmit and receive powers on RBs, including the inter-block interference, as follows:
\begin{align}
    2^{\mu_d}p_{z}^{rx}[d] = &\sum_{k \in \mathcal{K}}\sum_{l\in\mathcal{RB}_k} 2^{\mu_l}p_{k}^{tx}[l] s'_{(k,l),(z,d)} \notag \\
    & + \sum_{k \in \mathcal{K}}\sum_{l\in\mathcal{RB}_k} 2^{\mu_l}\Delta p_{k}^{tx}[l] q_{(k,l),(z,d)} + V_z[d],
\end{align}
where $\Delta p_{k}^{tx}[l] = p_{k}^{tx,(p)}[l] - p_{k}^{tx}[l]$, $s'_{(k,l),(z,d)} = s_{(k,l),(z,d)} + q_{(k,l),(z,d)}$, $p_{k}^{tx,(p)}[l]$ is the expected transmit power on the $l$-th RB of the previous OFDM symbol, and $q_{(k,l),(z,d)}$ is the new equivalent channel gain due to inter-block interference, which is proportional to TO, $\zeta_k$, and inversely proportional to the number of RBGs in a block, denoted as $n_{RBG}$. We derive the expression for $q_{(k,l),(z,d)}$ in a two-cell network of numerologies $i$ and $i'$ in Appendix~\ref{sec:q_vs_s}. For RBs $l$ and $d$ close in frequency band, $q_{(k,l),(z,d)}$ can be approximately written as
\begin{equation}
q_{(k,l),(z,d)} \approx \frac{N_L}{BN^i N^{i'}n_{RBG}}\sum_{r=B(l-1)+1}^{Bl}|H_{k,z}[r]|^2,
\end{equation}
where $H_{k,z}[r]=\sum_{l=0}^{L-1}h_{k,z}^{(l)}e^{\frac{-j2\pi rl}{N^i}}$ and $N_L = L+\zeta_k-N_{cp}^{i'}$. For strong interference channels, e.g., the co-channel interference between cells, the interference channel gain $s_{(k,d),(z,d)}$ can be as 
\begin{equation*}
    s_{(k,d),(z,d)} = \frac{N^i - \zeta_k}{BN^i N^{i'}}\sum_{r=B(l-1)+1}^{Bl}|H_{k,z}[r]|^2.
\end{equation*}
We can see that the ratio of $q_{(k,d),(z,d)}$ and $s_{(k,d),(z,d)}$ is very small for strong interference channels. However, the inter-block interference may still affect weak equivalent channel gains. In the power-domain, the impact of TO on channel estimation depends on the time fraction of TO to the block length. Since CP is typically 7\% of the OFDM symbol and a block includes a few OFDM symbols, the ratio of $q$ and $s$ is very small for nearby RBs even when TO exceeds CP, indicating that our approach is robust to TO. Fig.~\ref{fig:q_vs_s} shows the ratios of $q$ and $s$ for channels of different magnitudes in a three-cell network of two numerologies 15kHz and 30kHz, where each block includes two RBGs (see Section~\ref{sec:setup} for detailed setup) and $s_{max}$ is the strongest equivalent channel gain. We can see that even when TO is twice the CP length, the $q/s$ ratio is as low as 3\% for the interference channels 30dB less than $s_{max}$.



\begin{figure}[!t]
    \centering
        \includegraphics[scale=0.42]{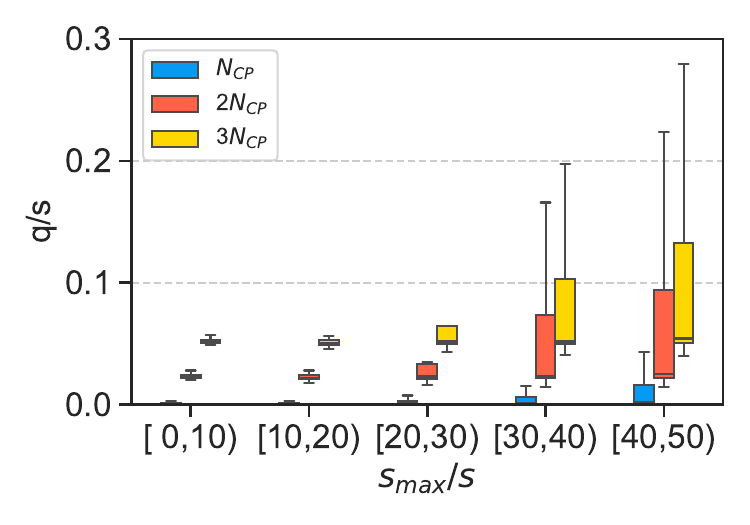}
    \caption{The ratio of equivalent channel gains $q$ and $s$}\label{fig:q_vs_s}
\end{figure}

\section{Joint Optimization for Interference Graph Estimation and Resource Allocation}
\label{sec:joint_opt}

In this section, we integrate IGE with a resource allocation problem aiming to maximize the energy efficiency of the multi-cell MN networks. We first present the formulation of the joint optimization problem and then propose a heuristic solution with provable convergence.

\subsection{Problem Formulation}
Our goal is to simultaneously 1) maximize the energy efficiency of the multi-cell MN system under the constraint that user traffic demand is satisfied and 2) minimize the condition number of the transmit power matrix for the accurate estimation of the interference graph. Let $\mathcal{U}_k$ be the set of UEs served by BS $k$, and $\mathcal{T}$ be the set of time blocks. We define network energy efficiency following the definition in~\cite{ee_def} as
\begin{equation*}
    E(\mathbf{P}) =\frac{f(\mathbf{P})}{g(\mathbf{P})} = \frac{\sum\limits_{k \in \mathcal{K}}\sum\limits_{z \in \mathcal{U}_k}\sum\limits_{d \in \mathcal{RB}_k}\sum\limits_{l \in \mathcal{T}}log(1+\text{SINR}_{k,z}[d][l])}{\sum\limits_{k \in \mathcal{K}}\sum\limits_{z \in \mathcal{U}_k}\sum\limits_{d \in \mathcal{RB}_k}\sum\limits_{l \in \mathcal{T}}p^{tx}_k[d][l]},
\end{equation*}
where $p_k^{tx}[d][l]$ is the transmit power of BS $k$ on RB $d$ at block $l$ and $\text{SINR}_{k,z}[d][l]$ is the signal-to-interference-plus-noise ratio of UE $z$ served by BS $k$ on RB $d$ at block $l$. Let $\delta_{k,z}[d][l]$ indicate if BS $k$ is sending to UE $z$ on RB $d$ at block $l$ and $p_{k,z}^{tx}[d][l]$ be the transmit power of BS $k$ to UE $z$ on RB $d$ at block $l$, where $p_{k,z}^{tx}[d][l] = 0$ if $\delta_{k,z}[d][l] = 0$. We have $p_k^{tx}[d][l] = \sum_{z\in\mathcal{U}_k}p^{tx}_{k,z}[d][l]$ and express the SINR as
\begin{align*}
    &\text{SINR}_{k,z}[d][l]=\notag\\
    &\frac{p_{k,z}^{tx}[d][l]s_{(k,d),(z,d)}}{\sum\limits_{\substack{k' \in \mathcal{K},\\d' \in \mathcal{RB}_{k'}}}  p_{k'}^{tx}[d'][l]s_{(k',d'),(z,d)}-p_{k,z}^{tx}[d][l]s_{(k,d),(z,d)}+V_{k,z}[d]},
\end{align*}
where $V_{k,z}[d]$ is the noise power on the $d$-th RB of UE $z$ under BS $k$. Based on the buffer status reports, we can determine user traffic demand and infer the required SINR requirements to satisfy user demands.

The joint optimization problem can be formulated as 
\begin{align}
(\mathcal{P}1) \ \max_{\mathbf{P}, \bm{\delta}} \ & \left[E(\mathbf{P}), -\sum_{k\in\mathcal{K}}\sum_{d\in\mathcal{RB}_k}\kappa(\mathbf{P}_k[d])\right] \notag \\
\textrm{s.t.} \ & (C_1)\ \sum\limits_{z \in \mathcal{U}_k}\delta_{k,z}[d][l] \leq 1,  \forall k, d, l, \notag \\
& (C_2)\ \sum\limits_{d \in \mathcal{RB}_k}\delta_{k,z}[d][l] \geq 1,  \forall k, z, l, \notag \\
& (C_3)\ \delta_{k,z}[d][l] \in \{0, 1\}, \forall k,z, \notag \\
& (C_4)\ \text{SINR}_{k,z}[d][l] \geq \gamma_{k,z} \delta_{k,z}[d][l], \forall k,z, \notag\\
& (C_{5})\ 0\leq p^{tx}_{k,z}[d][l] \leq \delta_{k,z}[d][l]P^{max}_k[d], \forall k,z, \notag\\
& (C_6)\ 0\leq \sum_{z\in \mathcal{U}_k}\sum_{d\in \mathcal{RB}_k} p^{tx}_{k,z}[d][l] \leq P^{max}_k, \forall k,l, \notag\\
&  (C_7)\ rank(\mathbf{P}_k[d]) = N_{RB} + 1, \forall k,d, \notag
\end{align}
where $\mathbf{P}_k[d] \in \mathbb{R}^{(N_{RB} + 1) \times (N_{RB} + 1)}$ is a transmit power matrix consisting of the transmit powers of the $d$-th RB of BS $k$ and its $N_{RB}$ interfering RBs, $\kappa(\mathbf{P}_k[d])$ is the condition number of $\mathbf{P}_k[d]$, $\delta_{(k, z)}[d][l]$ indicates if UE $z$ under BS $k$ is assigned RB $d$ at the $l$-th block, $\gamma_{k,z}$ is the required SINR for UE $z$ under BS $k$, $P^{max}_k[d]$ is the maximum transmit power for the $d$-th RB of BS $k$, $P^{max}_k$ is the total transmit power limit for BS $k$, $N_{RB}$ is the number of interfering RBs for each RB. 

\subsection{Proposed Solution}

\begin{figure}[!t]
    \centering
        \includegraphics[scale=0.45]{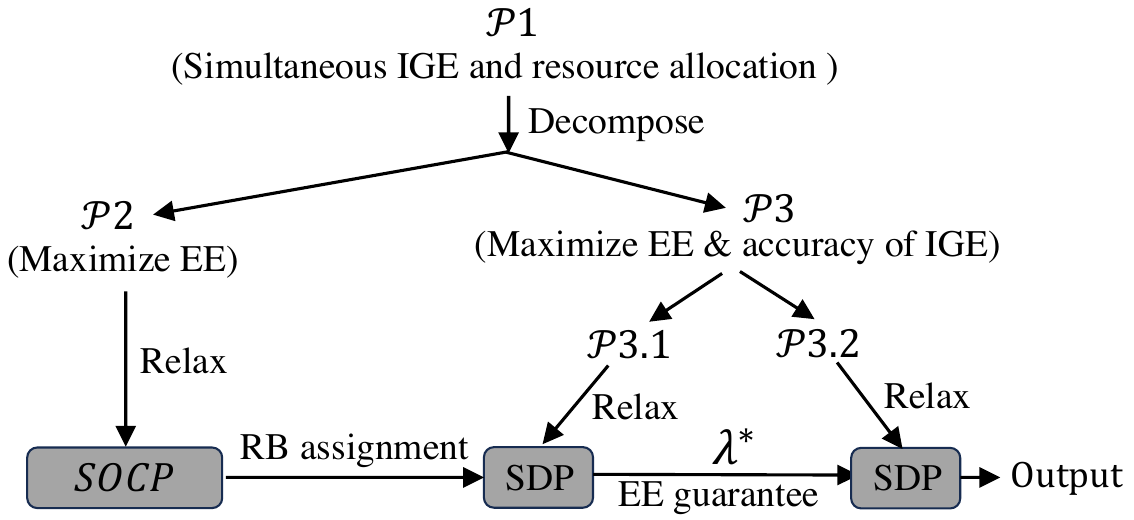}
    \caption{Workflow of our proposed solution}\label{fig:workflow}
    \vspace{-1em}
\end{figure}

$\mathcal{P}1$ is a multi-objective rank-constrained mixed-integer nonlinear programming (MINLP) problem. \textcolor{black}{Due to the nonlinearity of both the energy efficiency and condition number, combining them using the linear scalarization method would further complicate the problem. Instead, we use the $\epsilon$-constraint method to transform $\mathcal{P}1$ into coupled single-objective problems.}
As shown in Fig.~\ref{fig:workflow}, we first decompose $\mathcal{P}1$ into two sub-problems: an MINLP problem without the rank constraint ($\mathcal{P}2$) and a multi-objective problem with the rank constraints ($\mathcal{P}3$). $\mathcal{P}2$ aims to determine the assignment of RBs for each BS at different time blocks that maximizes the energy efficiency (EE), i.e., 
 \begin{align}
(\mathcal{P}2) \ \max\limits_{\mathbf{P}, \boldsymbol{\delta}} \ &E(\mathbf{P})\nonumber \\
\text{s.t.} \ & C_1\mbox{-}C_6, \nonumber
\end{align}
where $\boldsymbol{\delta}$ controls the assignment of RBs for UEs.
After that, $\mathcal{P}3$ takes the resulting RB assignment from $\mathcal{P}2$ as input and sequentially maximizes energy efficiency and the accuracy of IGE. $\mathcal{P}3$ is decomposed into two subproblems: 
\begin{align}
(\mathcal{P}3.1) \ \max\limits_{\mathbf{P}} \ & E(\mathbf{P}) \nonumber \\
\text{s.t.} \ & C_4 \mbox{-} C_7, \nonumber
\end{align}
and 
\begin{align}
(\mathcal{P}3.2) \ \min\limits_{\mathbf{P}} \ & \sum_{k\in\mathcal{K}}\sum_{d\in\mathcal{RB}_k}\kappa(\mathbf{P}_k[d]) \nonumber \\
\text{s.t.} \ & C_4 \mbox{-} C_7, \nonumber \\
& E(\mathbf{P}) \geq r\lambda^*, \nonumber
\end{align}
where $\lambda^*$ is the energy efficiency obtained by solving $\mathcal{P}3.1$ and is used in $\mathcal{P}3.2$ to provide the energy efficiency guarantee, and \textcolor{black}{$r\in (0, 1]$ is a tuning parameter to control the tradeoff between energy efficiency and condition number. 
We can solve $\mathcal{P}3.2$ under different $r$'s to obtain the Pareto front, as will be discussed in Section~\ref{sec:perf_design_param}.} Unlike $\mathcal{P}2$, $\mathcal{P}3.1$ and $\mathcal{P}3.2$ do not include binary variables and aim to optimize EE and IGE with power control. We will solve $\mathcal{P}2$ and $\mathcal{P}3$ in sequence with various relaxation and approximation techniques.

\subsection{Solution to $\mathcal{P}2$}
We want to use the Dinkelbach’s method~\cite{fractional_programming} to solve $\mathcal{P}2$, which requires the concavity of the objective function and convexity of the constraints. To satisfy these requirements, we need to transform the fractional objective function, relax the binary variables, and approximate the non-convex items.

\subsubsection{Transformation of the Fractional Objective}
We first transform the fractional objective into a subtractive form. After the transformation, the objective function becomes
\begin{align}\label{eq:objective_func}
\max\limits_{\bm{P}, \boldsymbol{\delta}} \ f(\mathbf{P}) - \lambda g(\mathbf{P}),
\end{align}
where $\lambda$ is initialized to a feasible solution for energy efficiency and updated after each iteration until convergence.
Since the Dinkelbach's method requires the concavity of $f(\mathbf{P})$, the convexity of $g(\mathbf{P})$, and the convexity of the feasible set, we need to solve the non-concavity issue in the objective function and the non-convexity issue in the constraints.

\subsubsection{Continuous Relaxation}
We eliminate the mixed-integer property of the original problem by converting the binary variables $\delta_{k,z}[d][l]$ into continuous variables between 0 and 1. To ease the conversion from the relaxed continuous variables back to binary numbers, we introduce a penalty term in the objective function  pushing $\delta_{k,z}[d][l]$ towards 0 or 1. After adding this penalty term, we rewrite objective function as
\begin{align}
\max\limits_{\bm{P}, \boldsymbol{\delta}} \ &f(\mathbf{P}) - \lambda g(\mathbf{P}) \notag \\ 
&+ \xi\sum_{k \in \mathcal{K}}\sum_{z \in \mathcal{U}_k}\sum_{d \in \mathcal{RB}_k}\sum_{l \in \mathcal{T}}\delta_{k,z}[d][l](\delta_{k,z}[d][l]-1). \notag
\end{align}


\subsubsection{Convex Approximation}
We want to solve the non-concavity issue in the objective function and the non-convexity issue in the constraint $C_4$. We rewrite $f(\mathbf{P})$ using the D.C. decomposition~\cite{dc_approx} as
\begin{equation*}
    f(\mathbf{P})=f_1(\mathbf{P})-f_2(\mathbf{P}),
\end{equation*}
where 
\begin{align*}
     &f_1(\mathbf{P}) = \sum_{k \in \mathcal{K}}\sum_{z \in \mathcal{U}_k}\sum_{d \in \mathcal{RB}_{k}}\sum_{l \in \mathcal{T}} \log(q(\mathbf{P})), \notag \\
    &f_2(\mathbf{P}) = \sum_{k \in \mathcal{K}}\sum_{z \in \mathcal{U}_k}\sum_{d \in \mathcal{RB}_k}\sum_{l \in \mathcal{T}} \log(q(\mathbf{P})
    -p_k[d][l]s_{(k,d), (z,d)}), \\
    &q(\mathbf{P}) = \sum_{k'\in \mathcal{K}}\sum_{d' \in \mathcal{RB}_{k'}}p_{k'}[d'][l]s_{(k',d'),(z,d)} + V.
\end{align*}
We define 
\begin{align*}
&f_1(\mathbf{P},\bm{\delta}) = f_1(\mathbf{P}) - \lambda g(\mathbf{P}) - \xi(\sum_{k \in \mathcal{K}}\sum_{z \in \mathcal{U}_k}\sum_{d \in \mathcal{RB}_k}\sum_{\beta \in \mathcal{T}}\delta_{kz}[d][l]),\\
&f_2(\mathbf{P},\bm{\delta}) = f_2(\mathbf{P}) - \xi(\sum_{k \in \mathcal{K}}\sum_{z \in \mathcal{U}_k}\sum_{d \in \mathcal{RB}_k}\sum_{l \in \mathcal{T}}\delta_{kz}[d][l]^2),
\end{align*}
and rewrite the objective function in Eq.~(\ref{eq:objective_func}) as
\begin{equation}\label{eq:transformed_obj}
\max\limits_{\boldsymbol{P}, \boldsymbol{\delta}} \ f_1(\textbf{P},\bm{\delta})- f_2(\textbf{P},\bm{\delta}).
\end{equation}
To make the objective function concave, $f_2(\mathbf{P},\bm{\delta})$ should be transformed to be affine by employing the first-order Taylor approximation, i.e.,
\begin{align}
\widetilde{f}_2(\textbf{P},\bm{\delta})\triangleq &f_2(\textbf{P}^{(t-1)},\bm{\delta}^{(t-1)}) \notag\\
&+\nabla_{f_2}^T(\textbf{P}^{(t-1)},\bm{\delta}^{(t-1)})[(\textbf{P},\bm{\delta})-(\textbf{P}^{(t-1)},\bm{\delta}^{(t-1)})]. \notag
\end{align}



To relax $C_4$, we first decompose it into two constraints, i.e.,
\begin{align*}
&(C_{4.1}) \ G_{k,z}[d][l] = \notag\\
&{\sum\limits_{k' \in \mathcal{K}, d' \in \mathcal{RB}_{k'}}  p_{k'}[d'][l]s_{(k',d'),(z,d)}-p_k[d][l]s_{(k,d),(z,d)}+V}
\end{align*}
and
\begin{equation*}
(C_{4.2}) \ \frac{p_k[d][l]}{\gamma_{k,z}} s_{(k,d),(z,d)} \geq \delta_{k,z}[d][l]G_{k,z}[d][l].
\end{equation*}
We use the successive convex approximation (SCA) to convert $C_{4.2}$ to
\begin{equation*}
(C_{4.3}) \ \frac{p_k[d][l]}{\gamma_{k,z}} s_{(k,d),(z,d)} \geq \frac{\phi_{k,z}[d][l]}{2}\delta_{k,z}[d][l]^2+\frac{G_{k,z}[d][l]^2}{2\phi_{k,z}[d][l]},
\end{equation*}
where $\phi_{k,z}[d][l]$ is a constant larger than 0 and the equality holds when $\phi_{k,z}[d][l]=\frac{G_{k,z}[d][l]}{\delta_{k,z}[d][l]}$.


Combining all the above operations, we can rewrite $\mathcal{P}2$ as
\begin{align}
(\mathcal{P}2^*) \ \max\limits_{\mathbf{P}, \bm{\delta}, \mathbf{G}} \ &f_1(\textbf{P},\bm{\delta})- \widetilde{f}_2(\textbf{P},\bm{\delta}) \notag \\
\text{s.t.} \ & C_1, C_2, C_{4.1},C_{4.3}, C_5, C_6, \notag \\
              & \delta_{k,z}[d][l] \in [0, 1], \forall k,z,d,l
\nonumber
\end{align}
where $\mathbf{G}$ represents for variables $G_{k,z}[d][l]$'s. It can be seen that $\mathcal{P}2^*$ is a second-order cone programming (SOCP) problem. The solution to the original problem $\mathcal{P}2$ can be achieved by iteratively solving $\mathcal{P}2^*$ until convergence. After $\mathcal{P}2^*$ is solved, we update $\lambda$ in $f_1(\mathbf{P}, \bm{\delta})$ and solve the updated $\mathcal{P}2^*$. This process is repeated until convergence.


\subsection{Solution to $\mathcal{P}3$}

\subsubsection{Solving $\mathcal{P}3.1$}
After $\boldsymbol{\delta}$ is determined, $C_4$ becomes a linear constraint and $\mathcal{P}3$ degenerates to a rank-constrained nonlinear programming problem. With the transformed objective in Eq.~(\ref{eq:transformed_obj}), we can rewrite $\mathcal{P}3$ as
\begin{align*}
(\mathcal{P}3.1) \ \max\limits_{\mathbf{P}} \ & \Delta f(\textbf{P}, \hat{\bm{\delta}}) \\
\text{s.t.} \ & C_4 \mbox{-} C_7, \notag 
\end{align*}
where $\Delta f(\textbf{P}, \hat{\bm{\delta}}) = f_1(\textbf{P},\hat{\bm{\delta}})- \widetilde{f}_2(\textbf{P}, \hat{\bm{\delta}})$.
As shown in \cite{b22}, rank-constrained problems can be gradually approximated with a sequence of semidefinite programming (SDP) problems. 
Let $n_b = N_{RB} + 1$ and $\mathbb{S}^{n_b}$ be the set of symmetric $n_b \times n_b$ matrices. We can have $rank(\mathbf{P}_k[d]) = n_b$ if and only if there exists $\mathbf{Z}_k[d] \in \mathbb{S}^{n_b}$ and $\mathbf{U}_k[d] \in \mathbb{S}^{2n_b}$ such that
\begin{equation*}
    rank(\mathbf{Z}_k[d]) = n_b, 
\end{equation*}
and
\begin{equation*}
    rank\left(\mathbf{U}_k[d] \right) \leqslant n_b,\ \mathbf{U}_k[d] = \left[ \begin{matrix}
 \mathbf{I}_{n_b} & \mathbf{P}^T_k[d] \\
 \mathbf{P}_k[d] & \mathbf{Z}_k[d] \\
  \end{matrix} \right].
\end{equation*}
It has been proved in~\cite{b22} that when $e = 0$, $rank(\mathbf{P}_k[d]) = n_b$ is equivalent to 
\begin{equation*}
    R_{k}[d]^{T}\mathbf{Z}_k[d]R_k[d] > 0, e\mathbf{I}_{n_b}-\mathbf{W}_k[d]^T\mathbf{U}_k[d]\mathbf{W}_k[d] \succeq 0,
\end{equation*}
where $\mathbf{U}_k[d] = \left[ \begin{matrix}
 \mathbf{I}_{n_b} & \mathbf{P}_k[d]^T \\
 \mathbf{P}_k[d] & \mathbf{Z}_k[d] \\
  \end{matrix} \right]$, $R_k[d] \in \mathbb{R}^{n_b}$ is the eigenvector corresponding to the smallest eigenvalue of $\mathbf{Z}_k[d]$, and $\mathbf{W}_k[d] \in \mathbb{R}^{2n_b \times n_b}$ are the eigenvectors corresponding to the $n_b$ smallest eigenvalues of $\mathbf{U}_k[d]$.



With the SDP approximations, we can solve $\mathcal{P}3.1$ by iteratively solving the following problem:
\begin{subequations}\label{eq:iterative_sdp}
\begin{align}
\max\limits_{\mathbf{P}^{(t)}, \mathbf{Z}^{(t)}, \bm{e}^{(t)}} &\Delta f(\textbf{P}, \hat{\bm{\delta}}) - w^{(t)}\sum_{k\in\mathcal{K}}\sum_{d\in\mathcal{RB}_k}e^{(t)}_k[d] \\
\text{s.t.}\quad & R^{(t-1)}_k[d]^T\mathbf{Z}^{(t)}_k[d]R^{(t-1)}_k[d] > 0 \label{eq:sdq_s}, \\
& e^{(t)}_k[d]\mathbf{I}_{n_b} - \mathbf{W}^{(t-1)}_k[d]^T\mathbf{U}^{(t)}_k[d]\mathbf{W}^{(t-1)}_k[d]\succeq 0,  \\
& 0 \leq e^{(t)}_k[d] \leq e^{(t-1)}_k[d], \label{eq:sdp_e} \\
& \mathbf{Z}^{(t)}_k[d] \succeq \boldsymbol{0}, \mathbf{U}^{(t)}_k[d] \succeq \boldsymbol{0}, \\
& C_4, C_5, C_6 
\end{align}
\end{subequations}
where $\mathbf{P}^{(t)}$ is the transmit power matrix at the $t$-th iteration, $\mathbf{W}^{(t-1)}_k[d]$ includes the eigenvectors corresponding to $n_b$ smallest eigenvalues of $\mathbf{U}^{(t-1)}_k[d]$, and $w^{(t)}$ is a weighting factor increasing with $t$. 
At the $0$-th iteration, the initial solution can be obtained by solving the problem without constraints (\ref{eq:sdq_s})-(\ref{eq:sdp_e}), and $e^{(0)}_k[d]$ is the $n_b$-th smallest eigenvalue of $\mathbf{U}^{(0)}_k[d]$.

\subsubsection{Solving $\mathcal{P}3.2$}
To address nonconvexity of $\kappa(\mathbf{P}_k[d])$'s in the objective function, we express it in the form of semi-definite matrices using singular values. Since $\kappa(\mathbf{P}_k[d]) = \left(\frac{\lambda_{max}(\mathbf{P}_k[d])}{\lambda_{min}(\mathbf{P}_k[d])}\right)^{1/2}$ is the square root of the ratio of the largest and smallest singular values, we can rewrite $\mathcal{P}3.2$ as
\begin{align}
(\mathcal{P}3.2^*) \ \min_{\mathbf{P},\gamma} \ & \sum_{k\in\mathcal{K}}\sum_{d\in\mathcal{RB}_k}\eta_k[d] \nonumber \\
\text{s.t.} \ & C_4\mbox{-}C_7, C_8, C_9 \nonumber
\end{align}
where $C_8$ constrains the ratio between the largest and smallest singular values as
\begin{equation*}
(C_8) \ \mu\mathbf{I}_{n_b}  \prec \mathbf{P}_k[d]^T\mathbf{P}_k[d] \prec \eta_k[d]\mu \mathbf{I}_{n_b}, \forall k, d,
\end{equation*}
and $C_9$ is the linearized constraint for EE guarantee as
\begin{equation*}
    (C_9) \ f_1(\mathbf{P}) - f_2(\mathbf{P}) \geq \lambda^* q(\mathbf{P}).
\end{equation*}
Since $\mathbf{P}_k[d]^T\mathbf{P}_k[d]$ is nonlinear, we need to further transform constraint $C_8$. We know that when the iterative approach in Eq.~(\ref{eq:iterative_sdp}) converges, $e^{(t)}_k[d]$ approaches to zero and $rank(\mathbf{U}^{(t)}_k[d]) \leq n_b$. In~\cite{b22}, it has been shown that when $rank(\mathbf{U}^{(t)}_k[d]) \leq n_b$, we have $rank(\mathbf{Z}_k[d] - \mathbf{P}_k[d]\mathbf{P}_k[d]^T) = 0$, i.e., $\mathbf{Z}_k[d] = \mathbf{P}_k[d]\mathbf{P}_k[d]^T$. Since the non-zero eigenvalues of $\mathbf{P}_k[d]^T\mathbf{P}_k[d]$ are also the eigenvalues of $\mathbf{P}_k[d]\mathbf{P}_k[d]^T$, we can replace $\mathbf{P}_k[d]^T\mathbf{P}_k[d]$ with $\mathbf{Z}_k[d]$ in constraint $C_8$. Let $\mathbf{Z}_k[d]_{(i,j)}$ and $\mathbf{P}_k[d]_{(i,j)}$ be the elements of index $(i,j)$ in $\mathbf{Z}_k[d]$ and $\mathbf{P}_k[d]$, respectively. To handle the nonlinearity of $\eta_k[d]\mu$ in $C_8$, we apply variable substitution as $\nu=1/\mu$, $\mathbf{\tilde{Z}}_k[d]_{(i,j)}=\mathbf{Z}_k[d]_{(i,j)}/\mu$, $\tilde{P}_{ij}=P_{ij}/\mu$ and $\tilde{e}^{(t)}=e^{(t)}/\mu$,  we can rewrite $C_8$ as
\begin{equation*}
   (C_8^*) \ 
 \mathbf{I}_{n_b} \prec \mathbf{\tilde{Z}}^{(t)}_k[d] \prec \eta_k[d]\mathbf{I}_{n_b}, \forall k, d.
\end{equation*}
Following the variable substitution, we rewrite the constraints $C_4, C_5, C_6, C_9$ as $C_4^*, C_5^*, C_6^*, C_9^*$. The problem $\mathcal{P}3.2$ can then be solved by iteratively solving the following problem:
\begin{subequations}\label{eq:iterative_sdp_p3}
\begin{align}
\min\limits_{\substack{\mathbf{\tilde{P}}^{(t)}, \mathbf{\tilde{Z}}^{(t)}, \\ \bm{\tilde{e}}^{(t)}, \bm{\eta}, \nu}} &  \sum_{k\in\mathcal{K}}\sum_{d\in\mathcal{RB}_k} \eta_k[d] + w^{(t)}\sum_{k\in\mathcal{K}}\sum_{d\in\mathcal{RB}_k} \tilde{e}^{(t)}_k[d]\\
\text{s.t.}\quad 
& C_4^*, C_5^*, C_6^*, C_8^*, C_9^* \\
&\tilde{R}^{(t-1)}_k[d]^{T}\mathbf{\tilde{Z}}^{t}_k[d]\tilde{R}^{(t-1)}_k[d] > 0, \\
& \tilde{e}^{(t)}_k[d]\mathbf{I}_{n_b} - \mathbf{\tilde{W}}^{(t-1)}_k[d]^T\mathbf{\tilde{U}}^{(t)}_k[d]\mathbf{\tilde{W}}^{(t-1)}_k[d]\succeq 0,  
\label{eq:eigenvalue} \\
& 0 \leq \tilde{e}^{(t)}_k[d] \leq \tilde{e}^{(t-1)}_k[d], \\
& \mathbf{\tilde{Z}}^{(t)}_k[d] \succeq \boldsymbol{0}, \mathbf{\tilde{U}}^{(t)}_k[d] \succeq \boldsymbol{0}, \nu > 0,
\end{align}
\end{subequations}
where $\mathbf{\tilde{U}}^{(t)}_k[d] = \left[ \begin{matrix}
 \nu\mathbf{I}_{n_b} & {\mathbf{\tilde{P}}^{(t)}}_k[d]^T \\
 \mathbf{\tilde{P}}^{(t)}_k[d] & \mathbf{\tilde{Z}}^{(t)}_k[d] \\
  \end{matrix} \right]$. $\tilde{R}^{(t-1)}[d]$ and $\mathbf{\tilde{W}}^{(t-1)}_k[d]$ are computed from $\mathbf{\tilde{Z}}^{(t-1)}_k[d]$ and $\mathbf{\tilde{U}}^{(t-1)}_k[d]$ following the same way as that used to compute $R^{(t-1)}$ and $\mathbf{W}^{(t-1)}_k[d]$ from $\mathbf{Z}^{(t-1)}_k[d]$ and $\mathbf{U}^{(t-1)}_k[d]$, respectively. 

\subsection{Convergence Analysis}
Our iterative approach to solve $\mathcal{P}2$ consists of an outer loop for iterative EE approximation and an inner loop for the convex SOCP problem. To ensure the convergence of the entire algorithm, it is necessary to prove the convergence of both loops. By adopting the convergence analysis in~\cite{multi_numerology_b5g}, it can be proved that the outer loop of $\mathcal{P}2^*$ converges to an optimal candidate solution $\lambda$. Let $(\boldsymbol{P}^{(t)},\boldsymbol{\delta}^{(t)},\boldsymbol{G}^{(t)})$ be the solution of $\mathcal{P}2^*$ at the $t$-th inner loop iteration, we can see that this solution is still feasible at the $(t+1)$-th inner loop iteration. Since the objective of $\mathcal{P}2^*$ is concave, we have that the inner loop of $\mathcal{P}2^*$ is convergent. We then want to prove the iterative approach for $\mathcal{P}3$ converges. Specifically, we prove that if there exists a solution at the $(t-1)$-th iteration, a solution will also exist at the $t$-th iteration. Let $(\mathbf{\tilde{U}}^{(t)}, \tilde{e}^{(t)})$ be the solution to Eq. (\ref{eq:iterative_sdp_p3}) at the $t$-th iteration. Then, $(\mathbf{\tilde{U}}^{(t)}, \tilde{e}_k^{(t)}[d])$ satisfies Eq.~(\ref{eq:eigenvalue}) for all $k$'s, i.e., $\tilde{e}_k^{(t)}[d] \geq \lambda_{k,1} =  \lambda_{\max}\left(\mathbf{\tilde{W}}^{(t-1)}_k[d]^T\mathbf{\tilde{U}}^{(t)}_k[d]\mathbf{\tilde{W}}^{(t-1)}_k[d]\right), \forall k$. If we can prove that $\tilde{e}^{(t)}_k[d]\mathbf{I}_{n_b} - \mathbf{\tilde{W}}^{(t)}_k[d]^T\mathbf{\tilde{U}}^{(t)}_k[d]\mathbf{\tilde{W}}^{(t)}_k[d]\succeq 0$ is true for all $k$'s, i.e., $\tilde{e}_k^{(t)} \geq \lambda_{\max}\left(\mathbf{\tilde{W}}^{(t)}_k[d]^T\mathbf{\tilde{U}}^{(t)}_k[d]\mathbf{\tilde{W}}^{(t)}_k[d]\right)$, Eq. (\ref{eq:iterative_sdp_p3}) will have at least a solution $(\mathbf{\tilde{U}}^{(t)}, \tilde{e}^{(t)})$ at the $(t+1)$-th iteration. Let the $n_b$-th smallest eigenvalue of $\mathbf{\tilde{U}}^{(t)}[d]$ be $\lambda_{k,2} =  \lambda_{\max}\left(\mathbf{\tilde{W}}^{(t)}_k[d]^T\mathbf{\tilde{U}}^{(t)}_k[d]\mathbf{\tilde{W}}^{(t)}_k[d]\right)$. We next prove that $\lambda_{k,1} \geq \lambda_{k,2}$. Recall that $\mathbf{W}_k^{(t)}[d]$ includes the eignvectors corresponding to the $n_b$ smallest eigenvalues of $\mathbf{U}_k^{(t)}[d]$. It has been shown in \cite{b22} that for $\mathbf{W}^T\mathbf{W} = \mathbf{I}_{n_b}$, $\lambda_{\max}\left(\mathbf{W}^T\mathbf{\tilde{U}}^{(t)}\mathbf{W}\right) \geq \lambda_{k,2}$. Since $\left(\mathbf{W}_k^{(t-1)}[d]\right)^T\mathbf{W}_k^{(t-1)}[d] = \mathbf{I}_{n_b}$, we have $\lambda_{k,1} \geq \lambda_{k,2}$. At the first iteration ($t = 1$), the pair $(\mathbf{\tilde{U}}^{(0)},\tilde{e}^{(0)})$ is a feasible solution, with $\tilde{e}^{(0)}$ being the $(N_{RB}+2)$-th largest eigenvalue of $\mathbf{\tilde{U}}^{(0)}$. Since Eq. (\ref{eq:iterative_sdp_p3}) is solvable at each iteration and $\tilde{e}_k^{(t)}[d] \leq \tilde{e}_k^{(t-1)}[d]$, the iterative approach converges when $\tilde{e}_k^{(t)}[d]$ approaches to zero.

\subsection{Complexity Analysis}

Let $Z_k$ and $D_k$ denote the total number of UEs and RBs in BS $k$, respectively, and $K$ be the total number of BSs. Suppose that we consider a scheduling period with $L$ RB blocks. For $\mathcal{P}2$, we have $3L\sum\limits_{k\in\mathcal{K}} Z_kD_k$ real-valued variables, $L\sum\limits_{k\in\mathcal{K}} Z_kD_k$ SOCP constraints with a dimension of 4, and $L(3\sum\limits_{k \in \mathcal{K}}Z_kD_k + \sum\limits_{k \in \mathcal{K}}D_k + \sum\limits_{k \in \mathcal{K}}Z_k + K)$ linear constraints. Consequently, following the complexity analysis in \cite{complexity}, the worst-case computational complexity of solving $\mathcal{P}2$ using an interior-point method can be expressed as $\mathcal{O}((L(4\sum\limits_{k \in \mathcal{K}}Z_kD_k + \sum\limits_{k \in \mathcal{K}}D_k + \sum\limits_{k \in \mathcal{K}}Z_k + K))^{0.5}(3L\sum\limits_{k\in\mathcal{K}} Z_kD_k)^2(L(7\sum\limits_{k \in \mathcal{K}}Z_kD_k + \sum\limits_{k \in \mathcal{K}}D_k + \sum\limits_{k \in \mathcal{K}}Z_k + K)))$.

We next discuss the computational cost of subproblems $\mathcal{P}3.1$ and $\mathcal{P}3.2$, which are transformed into SDP problems. The computational complexity of an existing SDP solver using the interior point method is on the order of 
$ O(m(n^2 m + n^3)) $ based on \cite{b22}, where $m$ is the number of linear constraints and $n$ is the dimension of the linear matrix inequality in the SDP problem. For subproblem $\mathcal{P}3.1$, we have $m = \sum\limits_{k \in \mathcal{K}}D_k + 2L\sum\limits_{k \in \mathcal{K}}Z_kD_k + KL$ and $n = 2n_b\sum_{k \in \mathcal{K}}D_k$. The complexity of solving $\mathcal{P}3.1$ is $\mathcal{O}((\sum\limits_{k \in \mathcal{K}}D_k + 2L\sum\limits_{k \in \mathcal{K}}Z_kD_k + KL)((2n_b\sum_{k \in \mathcal{K}}D_k)^2(\sum\limits_{k \in \mathcal{K}}D_k + 2L\sum\limits_{k \in \mathcal{K}}Z_kD_k + KL)+(2n_b\sum_{k \in \mathcal{K}}D_k)^3))$. For subproblem $\mathcal{P}3.2$, we have $m = \sum\limits_{k \in \mathcal{K}}D_k + 2L\sum\limits_{k \in \mathcal{K}}Z_kD_k + KL+1$ and $n = 2n_b\sum_{k \in \mathcal{K}}D_k$, the complexity becomes $\mathcal{O}((\sum\limits_{k \in \mathcal{K}}D_k + 2L\sum\limits_{k \in \mathcal{K}}Z_kD_k + KL+1)((3n_b\sum_{k \in \mathcal{K}}D_k)^2(\sum\limits_{k \in \mathcal{K}}D_k + 2L\sum\limits_{k \in \mathcal{K}}Z_kD_k + KL+1)+(3n_b\sum_{k \in \mathcal{K}}D_k)^3))$. Our iterative approaches solve all subproblems in polynomial time, which greatly reduces the computational cost.

\section{Performance Evaluation}
\label{sec:perf_eval}

\begin{table}[tbp]
\centering  
\caption{Simulation Parameters}
\label{table:param}
\begin{tabular}{cc|cc} 
\hline
\#BS & 3 to 6 & Cell radius & 250m \\
Time slot length & 1ms & Max. BS power & 46dBm\\
Carrier freq. ($f_c$) & 3.5GHz & Noise spectrum & -174dBm/Hz\\
Numerology $i=0,1,2$ & 15$\times2^i$kHz & Max. RB power & 30dBm \\
\hline
\vspace{-2em}
\end{tabular}
\end{table}

\subsection{Simulation Setup}
\label{sec:setup}

We evaluate the performance of our proposed joint optimization for IGE and resource allocation via extensive simulations. The simulation parameters are selected to simulate a typical 5G multi-cell multi-numerology system. Our experimental setup consists of 3 hexagonal cells, each with one BS located in the center. The cell radius is 250 meters and neighboring BSs are 500 meters apart. Each cell reuses the same frequency bands to increase spectrum efficiency and contains 3 randomly located UEs experiencing interference from neighboring cells. To align with the 3GPP standards~\cite{TS38211}, time is divided into time slots of 1ms. Numerology 0 uses the narrowest SCS of 15kHz and each time slot includes 14 OFDM symbols of numerology 0. The number of OFDM symbols in a time slot doubles for numerology 1 using an SCS of 30kHz. For each OFDM symbol, the CP takes 7\% of the symbol time. Each cell adopts a different frame structure using three numerologies, where the portions of RBs using numerologies 0, 1, and 2 are 1/2, 1/4, 1/4 for all BSs but their locations are different across BSs. BSs are synchronized with non-negligible errors. TO is caused by both the synchronization error and the delay difference between multiple paths, which could be larger than the CP. The CFOs between BSs and UEs range between -0.5 and 0.5. 

We use the UMi-Street Canyon NLoS model~\cite{TR38901} to calculate the path loss (PL) as $PL(d)=22.4+35.3\log_{10}(d)+21.3\log_{10}(f_c)$,
where $f_c$ is the carrier frequency and $d$ is the distance. Small-scale fading is considered with the channel gains following an exponential distribution with a unit mean. The  SINR requirements of UEs are chosen randomly and the modulation and coding schemes (MCSs) are determined based on their required SINRs. The maximum transmit powers for BS and RB are 46dBm and 30dBm, respectively. Each RB has 12 subcarriers, all of which use the same transmit power. The square $M$-QAM modulation is used in our experiments with $M$ $\in \{4, 16, 64, 256\}$. BSs randomly select modulated symbols from the constellation to send on each subcarrier.  The simulation parameters are listed in Table \ref{table:param}.

\begin{figure}[!t]
    \centering
        \includegraphics[scale=0.34]{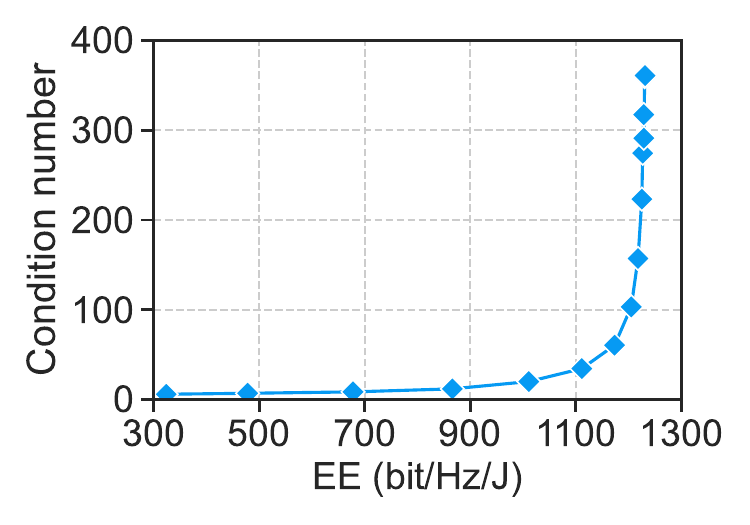}
    \caption{\textcolor{black}{Pareto front between EE and the condition number}}\label{fig:pareto}
\end{figure}

\begin{figure}[t]
    \centering
    \subfigure[Block length]{
        \label{fig:numerlogy_transmission_error}
        \includegraphics[scale=0.33]{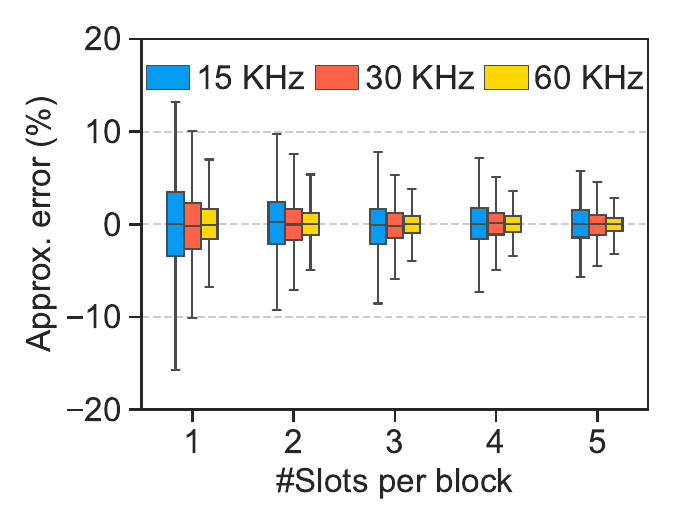}
    }
    \subfigure[Number of interfering RBs]{
        \label{fig:spectral_distance}
        \includegraphics[scale=0.33]{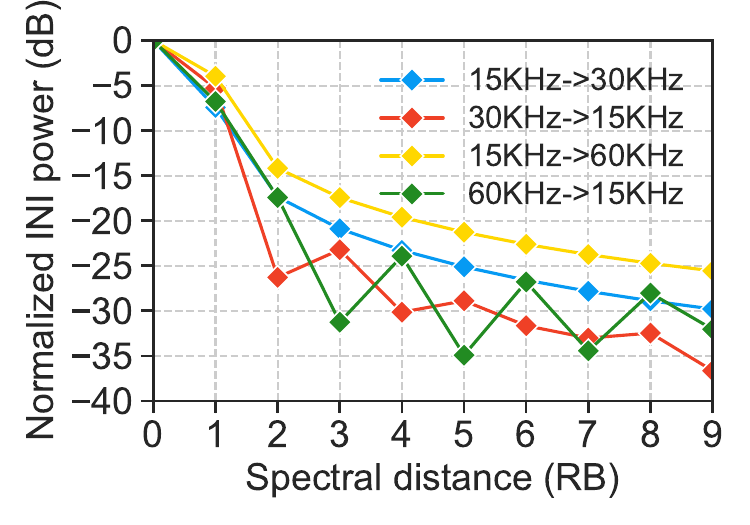}
    }
    \caption{Design parameter selection}
\end{figure}

\begin{figure}[!t]
    \centering
        \includegraphics[scale=0.34]{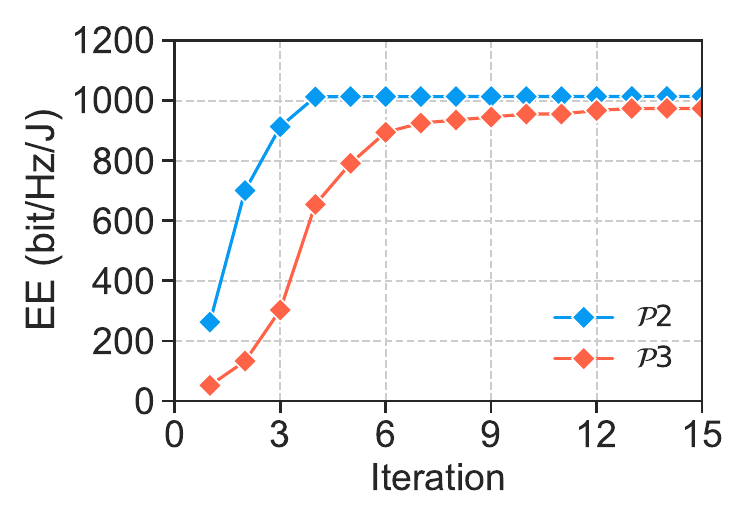}
    \caption{Convergence performance}\label{fig:convergence}
\end{figure}

 \begin{figure*}[t]
    \centering
    \begin{minipage}[t]{.3\textwidth}
        \vspace{0pt}
        \includegraphics[scale=0.36]{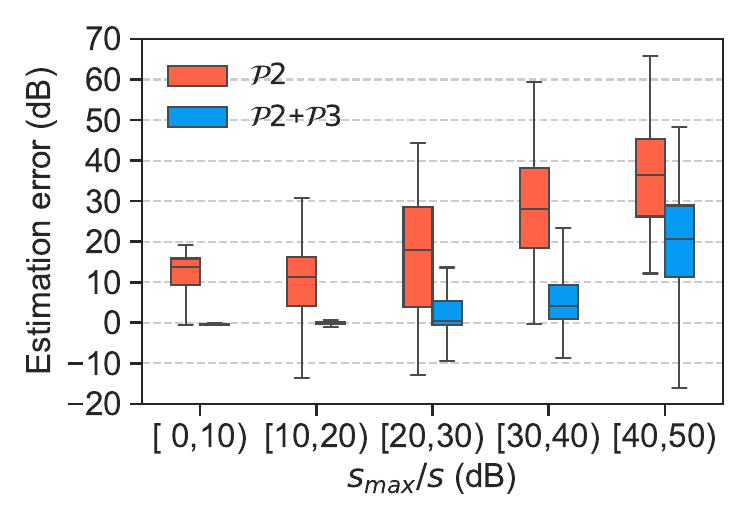}
        \caption{Estimation error}
        \label{fig:serror_vs_db}
    \end{minipage}
    \begin{minipage}[t]{.3\textwidth}
        \vspace{0pt}
        \includegraphics[scale=0.36]{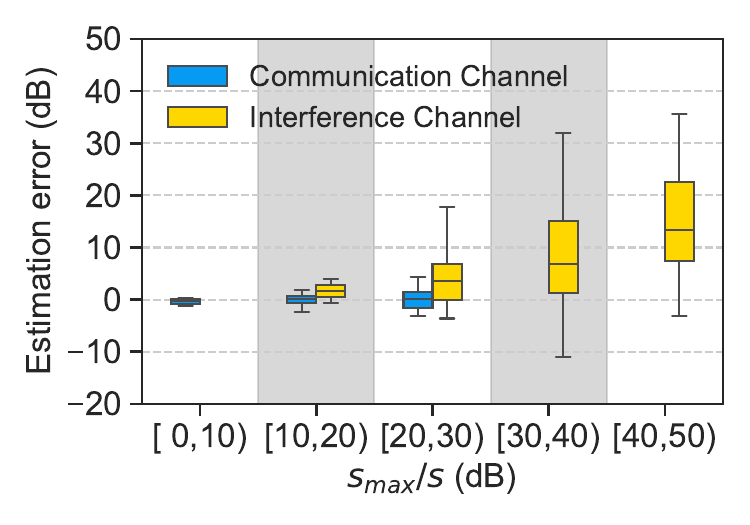}
    \caption{Interference decomposition}\label{fig:interference_link}
    \end{minipage}
    \begin{minipage}[t]{.3\textwidth}
        \vspace{0pt}
        \includegraphics[scale=0.36]{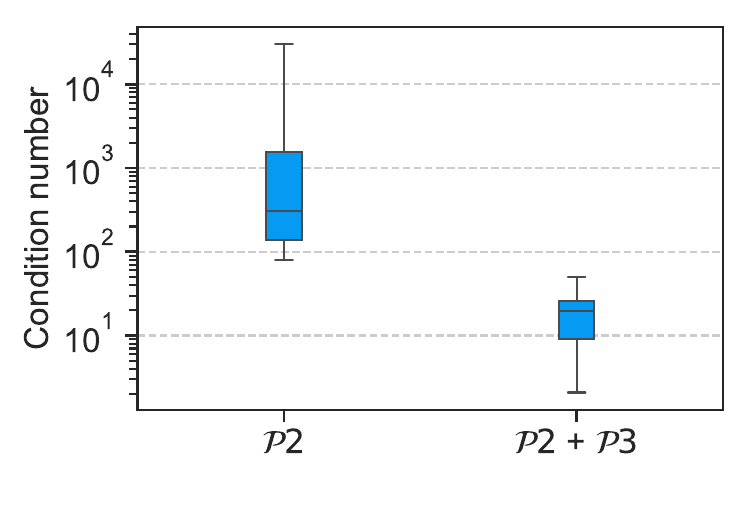}
    \caption{Impact of condition numbers}\label{fig:serror_vs_cn}
    \end{minipage}
\end{figure*}

\subsection{Design Parameter Selection and Convergence}
\label{sec:perf_design_param}


\noindent\textbf{\textcolor{black}{Tradeoff Between EE and the Condition Number.}}
\textcolor{black}{As a multi-objective optimization problem, improving EE would result in a larger condition number. Fig.~\ref{fig:pareto} shows the Pareto front between EE and the condition number. We observe that when EE is high, a small reduction in EE significantly lowers the condition number. Therefore, we prioritize optimizing EE and set $r=0.95$.}

\noindent\textbf{Block Length.}
As indicated by Theorem~\ref{theorem:bennet}, the actual transmit power can approximate the expected transmit power with a sufficient number of samples. However, using more samples per block would increase the time consumption for IGE, slowing down the frequency of interference graph updates. Let $\hat{\mu}_l$ and $\mu_l$ be the actual and expected transmit powers, respectively. Fig.~\ref{fig:numerlogy_transmission_error} shows the approximation error under different block lengths, where the block length is in the unit of slots (each slot = 1ms) and the approximation error is computed as $|\hat{\mu}_l - \mu_l| / \mu_l$. We can see that the approximation error decreases as the block length increases. When the block length is of two slots, the 75-th and 90-th percentiles are 2.5\% and 4.6\%, respectively. We thus set the block length to be of 2 slots.

\noindent\textbf{Scheduling Period Length.}
The scheduling period length is the product of block length and the number of blocks, where the number of blocks is the number of interfering RBs plus one. To speed up the IGE process, for each RB, we only select the RBs that could possibly incur significant interference to it. Since the boundary RBs between two numerologies experience the highest INI due to their proximity to the RBs of a different numerology, we focus on the INI experienced by the boundary RBs. Fig.~\ref{fig:spectral_distance} shows the relation between the INI power and the spectral distance, where the INI power is normalized with respect to the power of the boundary RB and the spectral distance is in the unit of RBs. The RBs with zero spectral distance are the boundary RBs. We can see that the INI power decreases with the spectral distance. When the spectral distance increases to two, the power of INI incurred to a boundary RB is 15 dB less than the power of the boundary RB itself. We therefore only consider the direct neighboring RBs for intra-cell INI. Considering that inter-cell interference is commonly weaker than the intra-cell INI due to the long distance between the interfering BS and the UE, we only consider the interference from the co-channel RBs. 


\noindent\textbf{Convergence.}
We evaluate the convergence performance of our proposed approach under random topologies, where UEs are randomly located. Fig. \ref{fig:convergence} shows the convergence process of our approach for both subproblems under a random topology. The SOCP approach for subproblem $\mathcal{P}2$ takes less than 4 iterations to converge and the SDP approach for subproblem $\mathcal{P}3$ takes about 14 iterations to converge. The SDP approach converges slower because $\mathcal{P}3$ solves two subproblems sequentially, with each subproblem including a group of rank constraints. These results indicate that our iterative approach is convergent. 

\subsection{Performance Evaluation}


\noindent\textbf{Power-Domain Channel Gain Estimation.}
Let $s_{(k',d'),(z,d)}$ and $\hat{s}_{(k',d'),(z,d)}$ denote the actual and estimated equivalent channel gains from RB $d'$ of BS $k'$ to RB $d$ of UE $z$, respectively. We use the logarithmic error, calculated as $10\log_{10}\left(\hat{s}_{(k',d'),(z,d)}/s_{(k',d'),(z,d)}\right)$ in dB, to measure the estimation error, where the error is zero when the estimated and actual channel gains are equal. We estimate the equivalent channel gains with Eq. (\ref{eq:least_square}) using the transmit power matrices obtained by solving $\mathcal{P}2$ alone and $\mathcal{P}2$+$\mathcal{P}3$, respectively. Recall that $\mathcal{P}2$ only optimizes the energy efficiency without considering IGE, while $\mathcal{P}2$+$\mathcal{P}3$ solves for both energy efficiency and IGE. Fig. \ref{fig:serror_vs_db} shows the estimation errors for equivalent channel gains of different magnitudes at the RB level, where the magnitude is relative to the maximum equivalent channel gain in the network, denoted as $s_{max}$. It can be seen that the estimation error increases for both $\mathcal{P}2$ and $\mathcal{P}2$+$\mathcal{P}3$ as the equivalent channel gain decreases, and that $\mathcal{P}2$+$\mathcal{P}3$ achieves significantly better estimation errors than $\mathcal{P}2$. This indicates the necessity of considering IGE while optimizing energy efficiency as $\mathcal{P}2$+$\mathcal{P}3$ does. Using the transmit power matrix from $\mathcal{P}2$+$\mathcal{P}3$, we can accurately estimate the interference channel gains 30dB less than $s_{max}$, where the median estimation error is less than 0.5dB. Although weak interference channels are likely to experience larger estimation errors, their influence on the overall interference estimation is very limited due to small channel gains. The interference experienced by UEs can thus be accurately estimated at the RB level. We further decompose the estimated channel gains into communication and interference channels. From Fig.~\ref{fig:interference_link}, we can see that communication channels are mostly strong enough to be accurately estimated, so are the strong interference channels. In contrast, weak interference channels tend to experience larger estimation errors but, as mentioned above, their impact on overall interference is limited.

\noindent\textbf{Impact of the Condition Number.} 
To better understand the objective of minimizing condition numbers in $\mathcal{P}3$, we compare the condition numbers before and after $\mathcal{P}3$ is solved. As shown in Fig.~\ref{fig:serror_vs_cn}, most condition numbers of the transmit power matrices range between 100 and 2000 before solving $\mathcal{P}3$, while after $\mathcal{P}3$ is solved, most condition numbers drop sharply below 30. Recall that the error bound of channel gain estimation is proportional to the condition number as in Eq. (\ref{eq:error_bound}). Decreasing the condition number helps reduce the channel gain estimation error, which aligns with our simulation results.

\begin{figure}[t]
    \centering
    \subfigure[Sync error = 0]{
        \label{fig:to_0}
        \includegraphics[scale=0.32]{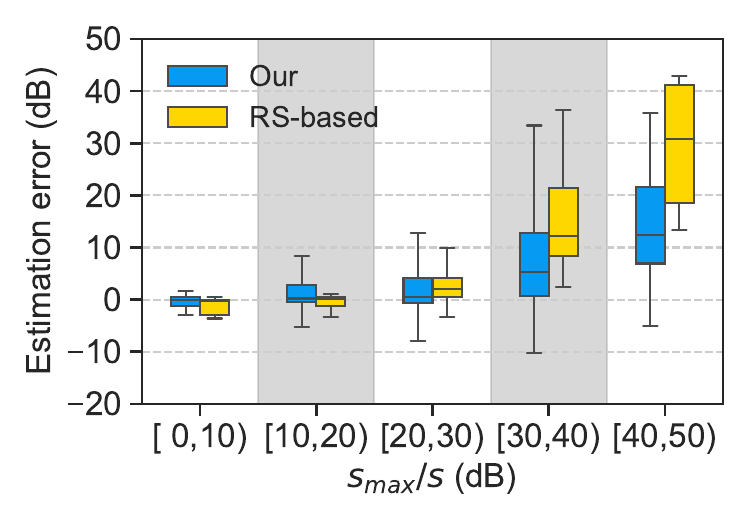}
    }
    \subfigure[Sync error = 2$N^1_{CP}$]{
        \label{fig:to_2}
        \includegraphics[scale=0.32]{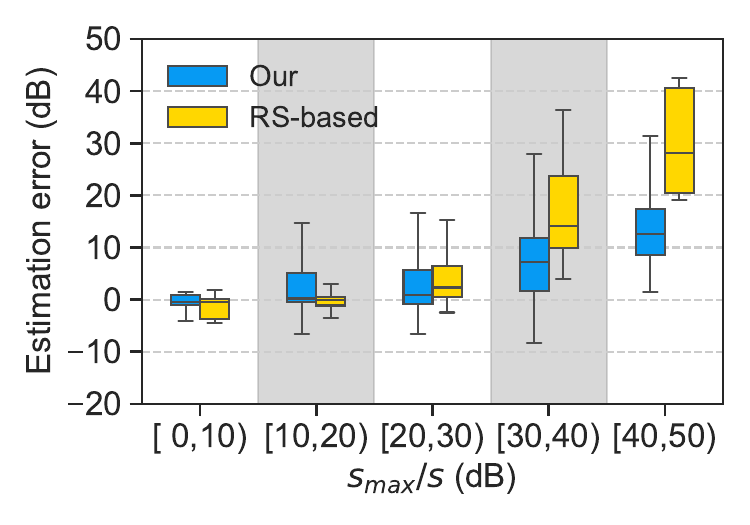}
    }
    \caption{Robustness to timing offset (TO)}
    \label{fig:error_vs_to}
\end{figure}

\noindent\textbf{Robustness to TO and CFO.}
We know that CP is designed to avoid inter-symbol interference within a cell. However, the synchronization error between BSs could cause TO to be greater than the CP length. Fig. \ref{fig:error_vs_to} shows the estimation errors under different synchronization errors, where the synchronization error is normalized by the CP length of numerology 1, denoted as $N^1_{CP}$. It can be seen that our approach can achieve very small estimation errors even when the synchronization error is twice the CP length. The robustness of our approach to TO is because we conduct IGE in the power domain. Moreover, since our approach measures channel gains in the power domain, it can work with multiple numerologies, while the traditional RS-based approach cannot effectively measure interference channels due to the misalignment in both the time and frequency between subcarriers of different numerologies. To demonstrate the robustness of our approach to CFO, we evaluate the estimation error under different CFOs. Fig.~\ref{fig:error_vs_cfo} shows that the estimation errors remain stable for the large channel gains and have a few decibels difference for weak channel gains. As the impact of weak channel gains is small, our approach is also robust to CFOs.



\begin{figure}[!t]
    \centering
    \includegraphics[scale=0.42]{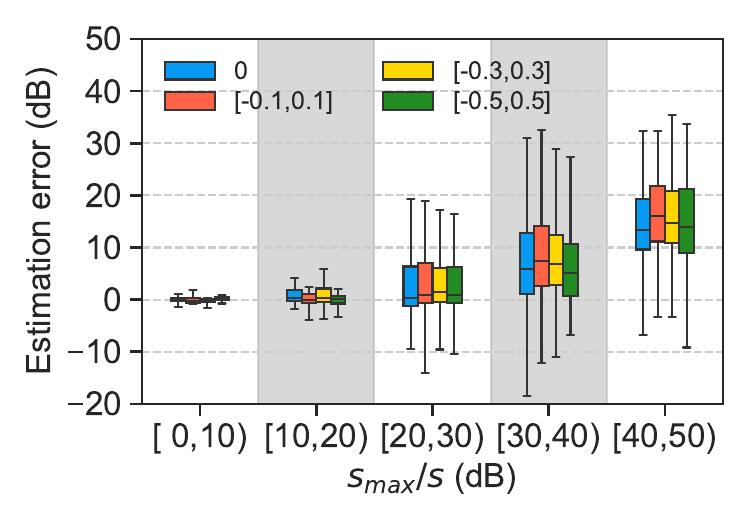}
    \caption{Robustness to CFO}
    \label{fig:error_vs_cfo}
\end{figure}

\noindent\textbf{Comparison with Existing Methods.}
Reference signals (RSs) are used in multi-numerology networks for channel estimation~\cite{rs_multinumerology}. Since the RS-based approach requires the transmitter and receiver to use the same numerology, it is limited to measure interference between RBs of the same numerology. For a fair comparison, we compare our approach and the RS-based approach only on the interference channel gains between RBs of the same numerology. Fig.~\ref{fig:error_vs_to} shows the estimation errors for channel gains of different magnitudes. Our approach outperforms the RS-based approach even when the synchronization error is zero because the reference signals experience significant INI and ICI from other subcarriers. In contrast, since our approach considers both INI and ICI channels, we can achieve higher accuracy for estimated channel gains. Our approach is robust to TO as can been seen in Fig.~\ref{fig:to_2}, where the estimation errors only increase slightly when the synchronization error is twice the CP length. The impact of TO on the RS-based approach is small because the power of INI dominates the power variations introduced by TO. We further compare our approach with the model-based approach that estimates interference with analytical expressions. We assume that the model-based approach can accurately estimate the path loss but is not adaptive to small-scale fading. Fig.~\ref{fig:model_based} shows the estimation errors between our and the model-based approaches. Since the major source of error for the model-based approach is small-scale fading, the estimation errors are consistently large for both strong and weak channel gains. Incorporating small-scale fading in model-based approaches may mitigate the estimation error, but no existing approaches are effective in measuring interference channel gains in multi-numerology networks.

\begin{figure}[!t]
    \centering
    \includegraphics[scale=0.42]{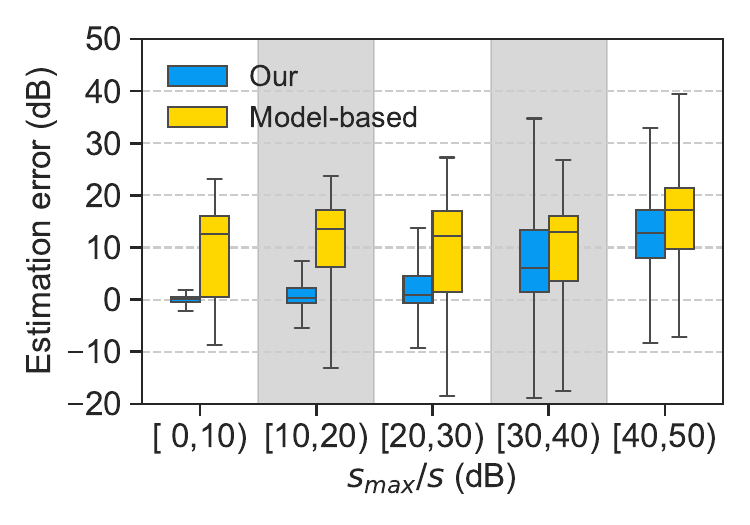}
    \caption{Comparison with the model-based approach}
    \label{fig:model_based}
\end{figure}

\noindent\textbf{Energy Efficiency Loss.}
As seen above, $\mathcal{P}3$ can optimize the condition number for the accuracy of IGE, but it also incurs energy efficiency loss. We want to evaluate energy efficiency loss under different network sizes, where the network topology is formed by one center BS surrounded by other BSs and the network size is controlled by the number of surrounded BSs. Fig.~\ref{fig:ee_decrease} shows that the median energy efficiency is about 5\% when the number of BSs is 3 and decreases to less than 3\% when there are 6 BSs, because the power overhead for IGE increases slower than the total power consumption of the network.

\section{Discussions and Limitations}

\noindent\textbf{Compatibility with the Existing Standards.}
In 3GPP standards~\cite{TR38521}, UEs may report channel status information (CSI) to BSs to improve channel quality, where the CSI is typically quantized to reduce reporting overhead. Our approach requires UEs to report the average receive power, which is not a typical CSI metric. Further, if quantized, the receive power may lose precision, affecting the accuracy of IGE. Although our approach saves frequency-time resources by performing interference channel estimation in the power domain, integrating it into standards would require adding new CSI metrics and choosing a proper quantization.

\noindent\textbf{Fine-Grained and Block-level Power Control.}
The transmit powers of BSs obtained by solving the joint optimization problem are continuous values, while the transmit power in real systems may be discretized for easy implementation. Moreover, constructing a full-rank matrix becomes more challenging when using discrete-valued transmit powers. Further, to approximate expected transmit power, we enforce RBs contiguous in time to be used for the same user. This aligns with the transmission time interval (TTI) that specifies the smallest unit of time to schedule a user~\cite{5g_slicing}. However, low-latency application may require the TTI to be a mini-slot, which is not long enough to well approximate the expected transmit power.

\noindent\textbf{Channel Gain Estimation for Resource Allocation.}
Our power-domain approach for IGE estimates only the magnitude of the channel gains. Consequently, these estimated channel gains cannot be used for demodulation. Nonetheless, as the channel gains are extensively used to estimate SINR and channel capacity, our power-domain approach for IGE is useful for resource allocation. In particular, our approach is capable of estimating the interference channel gains that are previously difficult to measure in multi-numerology networks.

\section{Conclusion}
\label{sec:conclusion}

In this paper, we proposed a power-domain approach to estimate the interference graph, which is crucial for the resource allocation in multi-cell multi-numerology systems. Since our approach estimates the interference graph in power domain, it is robust to timing and carrier frequency offsets. We derived the linear relation between transmit/receive power and interference channel gains. Based on this relation, we provided the necessary condition for the feasibility of interference graph estimation and designed a practical power control scheme. We formulated a joint optimization of interference graph estimation and resource allocation and proposed iterative solutions. Simulation results show that strong interference channel gains can be accurately estimated with low overhead. For future work, we want to extend the joint optimization framework to other wireless networks.

\begin{figure}[!t]
    \centering
    \includegraphics[scale=0.42]{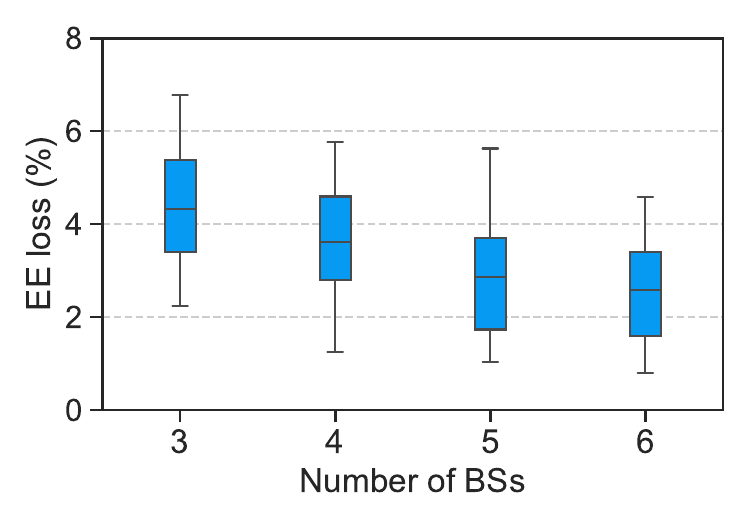}
    \caption{Energy efficiency loss due to IGE}
    \label{fig:ee_decrease}
\end{figure}


\appendices
\section{Proof of Lemma \ref{lemma:linearity_slow_moving}}
\label{section:linearity_slow_moving}

Without loss of generality, we consider a multi-cell network with two numerologies $i$ and $i'$, where numerology $i'$ has a larger SCS than numerology $i$ and $T_{i} = 2^{i'} T_{i'}$. Based on Eq. (\ref{eq:received_signal}) and $x_k[n]= \sum_{i\in\mathcal{I}_k} x^i_k[n]$, the frequency-domain received signal of UE $z$ using numerology $i$ can be expressed as
\begin{equation}\label{eq:dft_on_UE}
Y_z^i[d]=\frac{1}{\sqrt{N^i}}\sum_{k \in \mathcal{K}}\sum_{l=0}^{L-1}\sum_{n=0}^{N^i-1}h_{k,z}^{(l)}  \sum_{i\in\mathcal{I}_k} x^i_k[n-l-\zeta_k]e^{-j\frac{2\pi(\omega_k-d)n}{N^i}}.
\end{equation}
Let $X_{k,l}^{i}[d]$ be the symbol transmitted by BS $k$ on subcarrier $d$ of the $l$-th OFDM symbol in the least common
multiplier (LCM) duration. Since $\mathbb{E}[X_{k,l}^{i}[d]] = 0$ and $X_{k,l}^{i}[d]$'s are independent, we have that $\mathbb{E}[X_{k,l}^{i}[d]X_{k',l'}^{i'}[d']] \not= 0$ only if $i = i'$, $k = k'$, $l = l'$, and $d = d'$.

Let $\Delta_{d,m}^{i' \rightarrow i}=2^{i'} m-d$, and $\Delta_{d,m}^{i\rightarrow i'}=m-2^{i'} d$. From Eq. (\ref{eq:dft_on_UE}), we have that
\begin{align*}
2^{i'}\mathbb{E}[|Y_z^{i'}[d]|^2] = 
 &\sum_{k \in \mathcal{K}}\sum_{m=0}^{N^{i'}-1} 2^{i'}\mathbb{E}[|X_k^{i'}[m]|^2]s_{(k,m),(z,d)}^{i'\rightarrow i'}\\
 &+\sum_{k\in \mathcal{K}}\sum_{m=0}^{N^i-1}\mathbb{E}[|X_k^{i}[m]|^2]s_{(k,m),(z,d)}^{i\rightarrow i'}+\hat{V}_z[d],
\end{align*}
where 
\begin{align*}
&s_{(k,m),(z,d)}^{i' \rightarrow i'}=\frac{2^{i'}}{(N^{i'})^2}\left|\sum_{l=0}^{L-1}h_{k,z}^{(l)}\sum_{n=0}^{N^{i'}-1}e^{\frac{j2\pi[(m+\omega_k-d)n-lm]}{N^{i'}}}\right|^2, \\
&s_{(k,m),(z,d)}^{i \rightarrow i'}=\frac{2^{i'}}{N^i N^{i'}}\left|\sum_{l=0}^{L-1}h_{k,z}^{(l)}\sum_{n=0}^{N^{i'}-1}e^{\frac{j2\pi[(\Delta_{d,m}^{i \rightarrow i'}+2^{i'}\omega_k)n-lm]}{N^i}}\right|^2.
\end{align*}

Since $X^{i}_{k,l}[d]$'s are independent, we can easily extend the above derivation to cases with more than two numerologies and have that
\begin{equation*}
2^i\mathbb{E}[|Y_z^i[d]|^2] = 
 \sum_{k \in \mathcal{K}}\sum_{j \in \mathcal{I}_k}\sum_{l=0}^{N^j-1} 2^j\mathbb{E}[|X_k^j[l]|^2]s_{(k,l)(z,d)}^{j\rightarrow i}+\hat{V}_z[d].
\end{equation*}

\section{Proof of Theorem~\ref{theorem:spectral_distance}}
\label{sec:proof_spectral_distance}

Without loss of generality, we consider the case of two numerologies $i$ and $i'$ within a cell. Let $m$ and $d$ be the subcarrier indices of numerology $i$ and $i'$, where $i, i' \in \mathcal{I}$ and $i' > i$. We focus on the interference channels from numerology $i$ to $i'$ and express the spectral distance between subcarriers $m$ and $d$ as $\Delta_{d,m}^{i\rightarrow i'}=m-2^{i'} d$. As detailed in Appendix~\ref{section:linearity_slow_moving}, the equivalent channel gain $s_{(k,m),(z,d)}^{i \rightarrow i'}$ can be expressed as 
\begin{equation*}
s_{(k,m),(z,d)}^{i \rightarrow i'} = \frac{2^{i'}}{N^i N^{i'}} \left|\sum_{l=0}^{L-1}h_{k,z}^{(l)}\sum_{n=0}^{N^{i'}-1}e^{\frac{j2\pi\left[(\Delta_{d,m}^{i\rightarrow i'}+2^{i'}\omega_k)n-lm\right]}{N^i}}\right|^2.
\end{equation*}
According to~\cite{windowed_ofdm}, in weakly or mildly time-dispersive environment, we can approximately have
\begin{equation*}
s_{(k,m),(z,d)}^{i \rightarrow i'} \approx 2^{i'}\frac{\left|H_{k,z}[m]\right|^2}{N^i N^{i'}} \left|\frac{\sin(\frac{\pi (\Delta_{d,m}^{i \rightarrow i'}+2^{i'}\omega_k)N^{i'}}{N^i})}{\sin(\frac{\pi(\Delta_{d,m}^{i \rightarrow i'}+2^{i'}\omega_k)}{N^i})}\right|^2,
\end{equation*}
where $H_{k,z}[m]=\sum_{l=0}^{L-1}h_{k,z}^{(l)}e^{\frac{-j2\pi ml}{N^i}}$. Let $s_0 = s_{(k,m),(z,m)}^{i \rightarrow i'}$ be the equivalent channel gain when $\Delta_{d,m}^{i \rightarrow i'} = 0$. We have that 
\begin{align*}
s_0 \approx \frac{\left|H_{k,z}[m]\right|^2}{\pi^2}\left|\frac{\sin(\pi w_k)}{w_k}\right| \geq \frac{2\left|H_{k,z}[m]\right|^2}{\pi^2},
\end{align*}
where $w_k \in [-0.5, 0.5]$. When $\Delta_{d,m}^{i\rightarrow i'} > 0$, nearby subcarriers have a spectral distance much less than the entire frequency band, i.e., $\Delta_{d,m}^{i \rightarrow i'} \ll N^i$. Using the approxmation that $\sin(x) \approx x$, we rewrite $s_{(k,m),(z,d)}^{i \rightarrow i'}$ as
\begin{align*}
s_{(k,m),(z,d)}^{i \rightarrow i'} &\approx 2^{i'}\frac{\left|H_{k,z}[m]\right|^2N^i}{ \pi^2 N^{i'}} \left|\frac{\sin(\frac{\pi (\Delta_{d,m}^{i \rightarrow i'}+2^{i'}\omega_k)N^{i'}}{N^i})}{\Delta_{d,m}^{i \rightarrow i'}+2^{i'}\omega_k}\right|^2.
\end{align*}
We can approximately bound $s_{(k,m),(z,d)}^{i \rightarrow i'}$ as $\frac{\left|H_{k,z}[m]\right|^2}{\pi^2} \frac{2^{2i'}}{|\Delta_{d,m}^{i \rightarrow i'}|^2}$.
The strength drop with respect to the spectral distance can be expressed as $\frac{s_{(k,m),(z,d)}^{i \rightarrow i'}}{s_0} \leq \frac{2^{2i' - 1}}{|\Delta_{d,m}^{i \rightarrow i'}|^2}$.
When two subcarriers are one RB apart, the spectral distance is $B$, i.e., the number of subcarriers in a RB. We can have 
$s_{(k,m),(z,d)}^{i \rightarrow i'} / s_0 \leq \frac{1}{2B^2}$, which is inversely proportional to the square of the number of subcarriers in a RB. 

\begin{table*}[t]
\centering
\begin{align}\label{eq:s_to}
s_{(k,m),(z,d)}^{i \rightarrow i'}&=\frac{1}{N^i N^{i'}}
\left[\left|\sum_{l=0}^{L-1}h_{k,z}^{(l)}\sum_{n=l+\zeta_k-N_{cp}^{i'}}^{N^{i'}-1}e^{\frac{j2\pi[(\Delta_{d,m}^{(i,i')}+2^{i'}\omega_k)n-lm]}{N^i}}\right|^2 +(2^{i'}-1)\left|\sum_{l=0}^{L-1}h_{k,z}^{(l)}\sum_{n=0}^{N^{i'}-1}e^{\frac{j2\pi[(\Delta_{d,m}^{(i,i')}+2^{i'}\omega_k)n-lm]}{N^i}}\right|^2\right], \nonumber \\
&\approx \frac{|H_{k,z}[m]|^2}{N^i N^{i'}}\left(\left|\frac{\sin(\frac{\pi(\Delta_{d,m}^{(i,i')}+2^{i'}\omega_k)(N^{i'}-(l+\zeta_k)+N^{i'}_{CP})}{N^i})}{\sin(\frac{\pi(\Delta_{d,m}^{(i,i')}+2^{i'}\omega_k)}{N^i})}\right|^2 + (2^{i'}-1)\left|\frac{\sin(\pi(\frac{\Delta_{d,m}^{(i,i')}}{2^{i'}} + w_k))}{\sin(\frac{\pi(\Delta_{d,m}^{(i,i')}+2^{i'}\omega_k)}{N^i})}\right|^2\right)
\end{align}
\hrule
\end{table*}

\section{Impact of Timing Offset on Channel Gain Estimation}
\label{sec:q_vs_s}

Without loss of generality, we consider a two-cell network, where the equivalent channel gains are affected by both INI and ICI. When TO exceeds CP, the equivalent channel gain at the subcarrier level can be written as
\begin{align*}
&q_{(k,m),(z,d)}^{i\rightarrow i'}=\frac{1}{N^i N^{i'}}\left|\sum_{l=0}^{L-1}h_{k,z}^{(l)}\sum_{n=0}^{N_l-1}e^{\frac{j2\pi[(\Delta_{d,m}^{i\rightarrow i'}+2^{i'}\omega_k)n-lm]}{N^i}}\right|^2, \\
&q_{(k,m),(z,d)}^{i' \rightarrow i'}=\frac{1}{N^i N^{i'}}\left|\sum_{l=0}^{L-1}h_{k,z}^{(l)}e^{\frac{-j2\pi lm}{N^{i'}}}\frac{1-e^{\frac{j2\pi(m+\omega-d)N_l}{N^{i'}}}}{1-e^{\frac{j2\pi(m+\omega-d)}{N^{i'}}}}\right|^2,
\end{align*}
where $N_l = l+t_k-N_{cp}^{i'}$.
Following the derivation in Appendix~\ref{sec:proof_spectral_distance}, we can approximately express $q_{(k,m),(z,d)}^{i\rightarrow i'}$ as
\begin{equation}
\begin{aligned}
q_{(k,m),(z,d)}^{i\rightarrow i'} & \approx \frac{|H_{k,z}[m]|}{N^i N^{i'}}\left|\frac{\sin(\frac{(\pi(\Delta_{d,m}^{i \rightarrow i'}+2^{i'}\omega_k)N_l}{N^i})}{\sin(\frac{\pi((\Delta_{d,m}^{(i,i')}+2^{i'}\omega_k)}{N_i})} \right|^2 \nonumber \\
& \approx \frac{|H_{k,z}[m]|N_l}{N^i N^{i'}}.
\end{aligned}
\end{equation}
Similarly, the equivalent channel gains $s_{(k,m), (s,d)}^{i\rightarrow i'}$ can be written as in Eq.~(\ref{eq:s_to}). For co-channel interference channels ($\Delta_{d,m}^{i \rightarrow i'} = 0$), we can approximately have
\begin{equation*}
s_{(k,m), (s,d)}^{i\rightarrow i'} \approx  \frac{|H_{k,z}[m]|^2}{N^i N^{i'}}(N^i - \zeta_k).
\end{equation*}

\bibliographystyle{IEEEtran}
\bibliography{IEEEabrv}

\end{document}